\documentclass[preprint,3p,times]{elsarticle}
\RequirePackage{multirow,booktabs,subfigure,color,array,hhline,makecell}
\usepackage{amssymb}
\usepackage{amsmath}
\usepackage{graphicx}
\usepackage{subfigure}
\usepackage{amsthm}
\usepackage{mathrsfs}
\usepackage{indentfirst}
\usepackage[colorlinks,citecolor=blue,urlcolor=blue]{hyperref}
\allowdisplaybreaks
\usepackage[table]{xcolor}
\usepackage{tikz-network}
\usepackage{pgf}
\usepackage{tikz}
\usetikzlibrary{arrows, decorations.pathmorphing, backgrounds, positioning, fit, petri, automata}
\usetikzlibrary{shadows,arrows,positioning}
\RequirePackage{algorithm,algpseudocode}
\allowdisplaybreaks
\usepackage[final]{changes}

\usepackage{caption}

\newtheorem{thm}{Theorem}
\newtheorem{defin}{Definition}
\newtheorem{lem}{Lemma}
\newtheorem{assum}{Assumption}
\newtheorem{rem}{Remark}

\newtheorem{con}{Condition}

\allowdisplaybreaks[4]
\AtBeginDocument{%
	\providecommand\BibTeX{{%
			\normalfont B\kern-0.5em{\scshape i\kern-0.25em b}\kern-0.8em\TeX}}}

\journal{Information Sciences}
\begin{document}

\captionsetup[figure]{labelfont={bf},labelformat={default},labelsep=period,name={Fig.}}

\begin{frontmatter}
\title{Finding mixed memberships in categorical data}

\author[label1]{Huan Qing\corref{cor1}}
\ead{qinghuan@cqut.edu.cn\&qinghuan@u.nus.edu\&qinghuan07131995@163.com}
\cortext[cor1]{Corresponding author.}
\address[label1]{School of Economics and Finance, Lab of Financial Risk Intelligent Early Warning and Modern Governance, Chongqing University of Technology, Chongqing, 400054, China}
\begin{abstract}
Latent class analysis, a fundamental problem in categorical data analysis, often encounters overlapping latent classes that introduce further challenges. This paper presents a solution to this problem by focusing on finding latent mixed memberships of subjects in categorical data with polytomous responses. We employ the Grade of Membership (GoM) model, which assigns each subject a membership score in each latent class. To address this, we propose two efficient spectral algorithms for estimating these mixed memberships and other GoM parameters. Our algorithms are based on the singular value decomposition of a regularized Laplacian matrix. We establish their convergence rates under a mild condition on data sparsity. Additionally, we introduce a metric to evaluate the quality of estimated mixed memberships for real-world categorical data and determine the optimal number of latent classes based on this metric. Finally, we demonstrate the practicality of our methods through experiments on both computer-generated and real-world categorical datasets.
\end{abstract}
\begin{keyword}
Categorical data\sep Grade of Membership model\sep regularized Laplacian matrix\sep sparsity \sep fuzzy modularity
\end{keyword}
\end{frontmatter}
\section{Introduction}\label{sec1}
Categorical data, commonly collected in social science research, involves a collection of subjects, items, and subjects' responses to items \citep{sloane1996introduction,agresti2012categorical}. For instance, in psychological evaluations, individuals serve as subjects while diverse statements represent items. Similarly, in educational assessments, students are the subjects, and questions serve as items. Another example is MovieLens 100k \citep{kunegis2013konect}, a user-rating-movie dataset where subjects are individual users and items are distinct movies. In most categorical data, a subject belongs to multiple latent classes. For instance, in personality tests, an individual can exhibit traits such as conscientiousness and dominance simultaneously. Similarly, in political surveys, individuals may hold both liberal and conservative ideologies with varying degrees of emphasis.

For the simple case that each subject belongs to a single latent class, perhaps the most popular model is the latent class model (LCM) \citep{goodman1974exploratory}. The LCM models the response matrix $R$ of categorical data by assuming that $R$'s expectation is the product of a classification matrix and the transpose of an item parameter matrix under a Binomial distribution. LCM can be estimated by Bayesian inference approaches using Markov chain Monte Carlo (MCMC) \citep{garrett2000latent,white2014bayeslca,li2018bayesian}, maximum likelihood estimation methods \citep{bakk2016robustness,chen2022beyond,gu2023joint}, tensor-based methods \citep{zeng2023tensor}, and provably consistent spectral clustering algorithms \citep{qing2023latent}. These methods are not immediately applicable to estimate subjects' mixed memberships when a subject may belong to multiple latent classes.

In this paper, we work with the well-known Grade of Membership (GoM) model \citep{woodbury1978mathematical}. GoM generalizes the LCM by letting each subject be associated with a membership score such that each subject belongs to different latent classes with different weights (i.e., extents). GoM also assumes that $R$ is generated from a Binomial distribution and $R$'s expectation is the product of a membership matrix and the transpose of an item parameter matrix. The goal of latent mixed membership analysis is to consistently estimate each subject's membership score from the observed response matrix $R$ of categorical data under the GoM model. Prior work on estimating latent mixed memberships under the GoM model includes Bayesian inference using MCMC \citep{erosheva2002grade,erosheva2007describing,gormley2009grade,gu2023dimension} and joint maximum likelihood (JML) algorithm  \citep{sirt_3.13-194}. However, both MCMC and JML do not have theoretical guarantees, and as pointed out in \citep{chen2024spectral}, they are time-consuming and perform unsatisfactorily when they are applied to deal with large-scale categorical data with many subjects and items. Recently, \citep{chen2024spectral} proposes an efficient method with theoretical guarantees under GoM for categorical data with binary responses. Though numerical results in \citep{chen2024spectral} show that their algorithm outperforms JML significantly both in efficiency and accuracy, we observe that work in \citep{chen2024spectral} still faces some limitations: (a) algorithm proposed in \citep{chen2024spectral} does not work for categorical data with polytomous responses; (b) work in \citep{chen2024spectral} does not consider data sparsity which measures the number of zeros (i.e., no-responses) in $R$; (c) work in \citep{chen2024spectral} does not consider how to accurately infer the number of latent classes $K$ for response matrices generated from the GoM model. Finally, there is no metric to measure the quality of estimated mixed memberships for real-world categorical data for all aforementioned works under GoM. Based on these observations, we make the following contributions in this paper:
\begin{itemize}
  \item \added{We introduce two efficient spectral methods for estimating latent mixed memberships and other GoM parameters in categorical data with polytomous responses. Our algorithms leverage the leading left singular vectors of a regularized Laplacian matrix derived from the response matrix. The fundamental concept underlying our methods is the existence of an ideal simplex structure and an ideal cone structure within the variants of these singular vectors for the oracle case where the response matrix's expectation is known under the GoM model. Leveraging this insight, we employ two established vertex-hunting algorithms — the successive projection algorithm \citep{gillis2015semidefinite} and the SVM-cone algorithm \citep{mao2018overlapping} — to identify the index set of pure subjects defined in Section \ref{sec2}. Utilizing this index set, we estimate the GoM parameters according to Lemma \ref{SVDPopulationLtau}. We further establish the convergence rates of the estimated membership scores for each subject under a mild condition on data sparsity. Given that our algorithms rely primarily on a select few leading singular vectors, they offer significant computational advantages for large-scale categorical data. To our knowledge, this is the first work to develop spectral methods based on a regularized Laplacian matrix for estimating the GoM model in categorical data with polytomous responses.} \deleted{We propose two efficient spectral methods to estimate latent mixed memberships and other GoM parameters for categorical data with polytomous responses. Our algorithms are developed by using a few leading left singular vectors of a regularized Laplacian matrix computed from the response matrix. The core idea behind our methods is the observation that there exist an ideal simplex structure and an ideal cone structure in the variants of the leading left singular vectors for the oracle case with known the expectation of the response matrix under the GoM model. Based on this observation, two popular vertex-hunting algorithms successive projection algorithm \citep{gillis2015semidefinite} and SVM-cone algorithm \citep{mao2018overlapping} are applied to the variants of the leading left singular vectors to find the index set of pure subjects defined in Section \ref{sec2}. With the help of the index set, we can estimate the GoM parameters based on Lemma \ref{SVDPopulationLtau}. We then establish the convergence rates of the estimated membership score for each subject under a mild condition on data sparsity. Since our algorithms are designed based on a few leading singular vectors of a matrix, they have a great advantage in computational cost for large-scale categorical data. To the best of our knowledge, this is the first work to develop spectral methods based on a regularized Laplacian matrix to estimate the GoM model for categorical data with polytomous responses.}
  \item \added{We introduce a metric to assess the quality of estimated mixed memberships in real-world categorical data. This metric builds upon the concept of fuzzy modularity, a widely used approach in social network analysis for evaluating the quality of overlapping community detection \citep{nepusz2008fuzzy}. By employing this metric, we determine the optimal value of $K$ by selecting the one that maximizes the fuzzy modularity.} \deleted{We propose a metric to measure the quality of estimated mixed memberships for real-world categorical data. Our metric is an application of the fuzzy modularity \citep{nepusz2008fuzzy} which has been widely used to measure the quality of overlapping community detection in social network analysis. We then estimate $K$ by choosing the one that maximizes the fuzzy modularity.} To our knowledge, our metric is the first to measure the quality of estimated mixed memberships for categorical data with polytomous responses.
  \item We conduct substantial numerical studies to investigate the efficiencies and accuracies of our algorithms. We also apply our algorithms and metric to four real-world categorical data with meaningful results.
\end{itemize}

The remaining sections of this work are structured as follows. Section \ref{relatedwork} offers a comprehensive overview of related works. Section \ref{sec2} lays the foundation for our model setup. Section \ref{sec3} introduces two algorithms, while Section \ref{sec4} provides theoretical guarantees for these algorithms. Section \ref{sec5} delves into the concept of fuzzy modularity and outlines a method to estimate the number of latent classes. Sections \ref{sec6} and \ref{sec7} present our simulated and empirical studies, respectively. Section \ref{sec8} concludes this paper. All proofs are included in the Appendix.
\section{Related works}\label{relatedwork}
\added{Previous research on categorical data has mainly centered around two general categories: algorithmic methods, which aim to optimize criteria that reveal disparities in responses across every conceivable partition, and model-based methods, which employ statistical models to theoretically model categorical data.}\deleted{Previous research on categorical data has primarily focused on two broad categories: algorithmic methods that optimize criteria reflecting differences among responses across all possible partitions, and model-based methods that utilize statistical models to theoretically model categorical data.}

For algorithmic methods, when each subject belongs to a single latent class, the K-modes (KM) algorithm \citep{huang1998extensions} is the first conventional algorithm for categorical data that uses Hamming
distance to replace Euclidean distance in the K-means algorithm and it enables the clustering of categorical data in a fashion similar to K-means. When each subject belongs to multiple latent classes, the fuzzy K-modes (FKM) algorithm \citep{huang1999fuzzy} is one of the most popular clustering algorithms for clustering categorical data. Both KM and FKM are sensitive to the initial cluster centers which can drop their performances in clustering. To overcome this limitation, some extensions of KM and FKM are developed in recent years, to name a few, \citep{cao2009new,khan2013cluster,jiang2016initialization,kuo2021metaheuristic,oskouei2021fkmawcw,xie2022dp,bai2022categorical}. However, these algorithmic methods do not have any theoretical guarantees of convergence rates (i.e., error rates) and they can not estimate the item parameter matrix $\Theta$ used for the LCM model and the GoM model in Equation (\ref{GoM}).

For model-based methods, latent class analysis (LCA) is a widely used statistical technique that identifies latent subgroups in categorical data across various fields, including machine learning, statistical analysis, social science, and psychology \citep{hagenaars2002applied, lanza2013latent, lanza2016latent,nylund2018ten,weller2020latent}. It is a method that groups subjects with similar responses in categorical data together \citep{nylund2018ten}. Compared to traditional clustering methods like K-means and K-modes, LCA offers a model-based clustering approach. A significant advantage of LCA over K-means and K-modes is that the selection of the clustering criterion is based on rigorous statistical tests. Consequently, LCA is considered more objective than other clustering methods \citep{He2018}. The LCM and GoM models are two popular statistical models used in LCA for categorical data. The LCM model assumes that each subject belongs to a single latent class, while the GoM model allows each subject to belong to multiple latent classes. As we mentioned earlier, though Bayesian inference and maximum likelihood estimation can be used to estimate LCM and GoM models, these methods lack theoretical guarantees and are computationally expensive \citep{chen2024spectral}. To develop methods with theoretical guarantees for categorical data with polytomous responses under the GoM model, this paper focuses on spectral clustering.

Spectral clustering approaches, relying on eigendecomposition or singular value decomposition of matrices, are widely utilized in model-based methods due to their theoretical soundness and ease of implementation \citep{ng2001spectral,von2007tutorial}. For example, in the realm of community detection, numerous spectral clustering methods have been developed under distinct statistical models for diverse types of networks \citep{rohe2016co,mao2018overlapping,mao2021estimating,qing2023regularized,guo2023randomized,qing2024bipartite}. Notably, spectral methods in \citep{mao2018overlapping,mao2021estimating,qing2023regularized,qing2024bipartite} are capable of estimating mixed memberships when nodes in networks belong to multiple communities. However, except for the spectral method introduced in \citep{chen2024spectral}, the concept of spectral clustering is seldom employed in LCA for categorical data with polytomous responses under the GoM model. To accommodate the GoM model for categorical data with polytomous responses, we present two novel spectral methods: Algorithm \ref{alg:SRSC} and \ref{alg:CRSC}, which are based on the regularized Laplacian matrix introduced in \citep{qing2023latent} and the vertex hunting techniques developed in \citep{mao2018overlapping, mao2021estimating}. Furthermore, we account for data sparsity and provide per-subject error rates for mixed memberships while theoretical analysis in \citep{chen2024spectral} ignores data sparsity and only offers overall error bounds for the estimated mixed membership matrix.

Meanwhile, the Newman-Girvan modularity \citep{newman2004finding,newman2006modularity} is a popular metric to measure the quality of community partition when the ground-truth communities are unknown and each node belongs to a single community. The fuzzy modularity introduced in \citep{nepusz2008fuzzy} extends the Newman-Girvan modularity to measure the quality of mixed membership community detection when nodes can belong to multiple communities. The Newman-Girvan modularity is used in \citep{qing2023latent} to measure the quality of latent class analysis when each subject belongs to a single latent class under the LCM model. However, to our knowledge, when the ground-truth mixed membership matrix is unknown for categorical data, there is no prior work to measure the quality of latent mixed membership analysis for categorical data. Compared to the model-based works in \citep{chen2024spectral,qing2023latent} and algorithmic-based works in \citep{huang1998extensions,huang1999fuzzy,cao2009new,khan2013cluster,jiang2016initialization,kuo2021metaheuristic,oskouei2021fkmawcw,xie2022dp}, the last novelty of this paper is that we use the fuzzy modularity in \citep{nepusz2008fuzzy} to measure the quality of latent mixed membership analysis and estimate the number of latent classes based on this modularity under the GoM model for categorical data with polytomous responses.
\section{The Grade of Membership (GoM) model}\label{sec2}
\begin{table}[ht]
\centering
\scriptsize
\rowcolors{1}{white!22}{white!22}
\resizebox{\columnwidth}{!}{
\begin{tabular}{cc|cc}
\hline
Symbol&Description&Symbol&Description\\
\hline
$N$&Number of subjects&$J$&Number of items\\
$M$&Positive integer&$R\in\{0,1,\ldots,M\}^{N\times J}$&Response matrix\\
$K$&Number of latent classes&$\Pi\in[0,1]^{N\times K}$&Membership matrix\\
$X(i,:)$&$i$-th row of any matrix $X$&$[m]$&$\{1,2,\ldots, m\}$ for any positive integer $m$\\
$\Theta\in[0,M]^{J\times K}$&Item parameter matrix&$\mathcal{R}\in[0,M]^{N\times J}$&$R$'s expectation matrix $\Pi\Theta'$\\
$\mathbb{E}(X)$&Expectation of $X$&$\mathbb{P}(x=a)$&Probability that $x$ equals to $a$ for any value $x$\\
$\mathcal{I}$&Index set of pure subjects&$X(S,:)$&Submatrix formed by rows in index set $S$ for any matrix $X$\\
$I_{m\times m}$&$m\times m$ identity matrix&$\tau$&Regularization parameter\\
$\mathscr{D}$&Diagonal matrix with $\mathscr{D}(i,i)=\sum_{j\in[J]}\mathscr{R}(i,j)$ for $i\in[N]$&$\mathscr{D}_{\tau}$&$\mathscr{D}+\tau I_{N\times N}$\\
$\mathscr{L}_{\tau}$&Population regularized Laplacian matrix $\mathscr{D}^{-1/2}_{\tau}\mathscr{R}$&$\|x\|_{q}$&$l_{q}$-norm for any vector $x$\\
$\mathrm{diag}(x)$&Diagonal matrix with $x$'s entries being its diagonal&$X(:,j)$&$j$-th column for any matrix $X$\\
$\|X\|_{F}$&Frobenius norm for any matrix $X$&$\sigma_{k}(X)$&$k$-th largest singular value for any matrix $X$\\
$U\in\mathbb{R}^{N\times K}$&Top $K$ left singular vectors of $\mathscr{L}_{\tau}$&$\Sigma\in\mathbb{R}^{K\times K}_{+}$&Diagonal matrix of top $K$ singular values of $\mathscr{L}_{\tau}$\\
$V\in\mathbb{R}^{J\times K}$&Top $K$ right singular vectors of $\mathscr{L}_{\tau}$&$U_{\tau}$&$\mathscr{D}^{1/2}_{\tau}U$\\
$D_{U}$&Diagonal matrix with $D_{U}(i,i)=\frac{1}{\|U(i,:)\|_{F}}$ for $i\in[N]$&$U_{*}$&$D_{U}U$\\
$X'$&Transpose of any matrix $X$&$\|X\|_{2\rightarrow\infty}$&$\mathrm{max}_{i}\|X(i,:)\|_{2}$ for any matrix $X$\\
$D_{o}$&Diagonal matrix with $D_{o}(i,i)=\frac{1}{\|\Pi(i,:)U_{\tau}(\mathcal{I},:)\|_{F}}$ for $i\in[N]$&$Y$&$D_{o}\Pi\mathscr{D}^{1/2}_{\tau}(\mathcal{I},\mathcal{I})D^{-1}_{U}(\mathcal{I},\mathcal{I})$\\
$D$&Diagonal matrix with $D(i,i)=\sum_{j\in[J]}R(i,j)$ for $i\in[N]$&$D_{\tau}$&$D+\tau I_{N\times N}$\\
$L_{\tau}$&Regularized Laplacian matrix $D^{-1/2}_{\tau}R$&$\hat{U}\in\mathbb{R}^{N\times K}$&Top $K$ left singular vectors of $L_{\tau}$\\
$\hat{\Sigma}\in\mathbb{R}^{K\times K}_{+}$&Diagonal matrix of top $K$ singular values of $L_{\tau}$&$\hat{V}\in\mathbb{R}^{J\times K}$&Top $K$ right singular vectors of $L_{\tau}$\\
$\hat{U}_{\tau}$&$D^{1/2}_{\tau}\hat{U}$&$\hat{\mathcal{I}}$&Estimated index set\\
$\hat{\Pi}$&Estimated membership matrix&$\hat{\Theta}$&Estimated item parameter matrix\\
$D_{\hat{U}}$&Diagonal matrix with $D_{\hat{U}}(i,i)=\frac{1}{\|\hat{U}(i,:)\|_{F}}$ for $i\in[N]$&$\hat{U}_{*}$&$D_{\hat{U}}\hat{U}$\\
$\rho$&Sparsity parameter $\mathrm{max}_{j\in[J],k\in[K]}\Theta(j,k)$&$B$&$\frac{\Theta}{\rho}$\\
$\pi_{\mathrm{min}}$&$\mathrm{min}_{k\in[K]}\sum_{i\in[N]}\Pi(i,k)$&$X^{-1}$&Inverse of any nonsingular matrix $X$\\
$\mathcal{P}$&$K\times K$ permutation matrix&$A$&Adjacency matrix defined as $RR'$\\
$d_{i}$&$\sum_{j\in[N]}A(i,j)$&$\omega$&$\sum_{i\in[N]}d_{i}$\\
$\lambda_{k}(X)$&$k$-th largest eigenvalue in magnitude for any matrix $X$&$\kappa(X)$&Condition number for any matrix $X$\\
$\mathbf{1}$&A vector with all elements being 1&$e_{i}$&$e_{i}(j)=1(i=j)$\\
$\mu$&Proportion of highly pure subjects&$\nu$&Proportion of highly mixed subjects\\
$\|X\|$&Spectral norm for any matrix $X$&$\delta_{\mathrm{min}}$&$\mathrm{min}_{i\in[N]}\mathscr{D}(i,i)$\\
$\delta_{\mathrm{max}}$&$\mathrm{max}_{i\in[N]}\mathscr{D}(i,i)$&$a=O(b)$&$a$ is of the same order as $b$\\
\hline
\end{tabular}
}
\caption{Main symbols used in this paper.}
\label{table-symbol}
\end{table}

This paper focuses on categorical data with polytomous responses, which are widely available in the field of sociology, psychology, and education, including never true/rarely true/sometimes true/often true/always true in psychological tests and a/b/c/d choices in educational assessments. For polytomous responses, as a convention, we use $\{0, 1, \ldots, M\}$ to denote different kinds of responses, where $0$ usually means no response and $M$ is a positive integer at least 1. Categorical data with polytomous responses can be mathematically represented by an $N\times J$ response matrix $R$ such that $R(i,j)$ takes values in $\{0,1,\ldots, M\}$ and it records the observed responses of the $i$-th subject to the $j$-th item, where $N$ denotes the number of subjects and $J$ represents the number of items. We present a list of main symbols in this paper in Table \ref{table-symbol}.

The GoM model assumes that all subjects belong to $K$ latent classes, where $K$ is assumed to be a known positive integer that is much smaller than $N$ and $J$ in this paper. Let $\Pi(i,:)$ be a $1\times K$ membership vector that satisfies $\Pi(i,:)\geq0$ and $\sum_{k=1}^{K}\Pi(i,k)=1$ for $i\in[N]$ and $k\in[K]$, where $\Pi(i,k)$ denotes the weight (probability) that the $i$-th subject belongs to the $k$-th latent class, $[m]=\{1,2,\ldots,m\}$ for any non-negative integer $m$, and $X(i,:)$ denotes $X$'s $i$-th row for any matrix $X$ in this paper. We call subject $i$ a pure subject if there is one entry of $\Pi(i,:)$ being 1 and call subject $i$ a mixed subject otherwise, i.e., a pure subject solely belongs to one latent class while a mixed subject belongs to multiple latent classes.

\added{Let $\Theta$ be a $J \times K$ item parameter matrix with elements ranging in $[0, M]$. The GoM model, designed to generate the response matrix $R$ for categorical data with polytomous responses, is formulated as follows:}\deleted{Let $\Theta$ be the $J\times K$ item parameter matrix whose elements range in $[0,M]$. The GoM model for generating the response matrix $R$ of categorical data with polytomous responses is as follows:}
\begin{align}\label{GoM}
\mathscr{R}:=\Pi\Theta'~~~~R(i,j)\sim \mathrm{Binomial}(M,\frac{\mathscr{R}(i,j)}{M})~~~~i\in[N], j\in[J],
\end{align}
where $R(i,j)\sim \mathrm{Binomial}(M,\frac{\mathscr{R}(i,j)}{M})$ \added{indicates}\deleted{means} that $R(i,j)$ is a random number generated from a Binomial distribution with $M$ \deleted{repeated} trials and \added{a} success probability \added{of} $\frac{\mathscr{R}(i,j)}{M}$. From Equation (\ref{GoM}), it is \added{evident} \deleted{clear to observe} that \added{the GoM model} \deleted{GoM} is \added{constituted} \deleted{formed} by the two model parameters\added{,} $\Pi$ and $\Theta$. For \added{brevity}\deleted{convenience}, we \deleted{denote the GoM model using the notation}\added{refer to the GoM model as} $\mathrm{GoM}(\Pi,\Theta)$. Equation (\ref{GoM}) \added{further} \deleted{also} implies that $R(i,j)$'s expectation under GoM is $\mathscr{R}(i,j)$, i.e., $\mathbb{E}(R)=\mathscr{R}$ under GoM. \added{Therefore, we designate $\mathscr{R}$ as the population response matrix. Furthermore, in scenarios where all subjects are pure (i.e., no mixed subjects), the GoM model simplifies to the widely used latent class model for categorical data.} \deleted{ Thus, we call $\mathscr{R}$ the population response matrix. What's more, when all subjects are pure (i.e., no mixed subjects), the GoM model reduces to the popular latent class model for categorical data.}

Since Binomial distribution is a discrete distribution, the probability of $R(i,j)$ equals choice $m$ is:
\begin{align}\label{GoMBinomial}
\mathbb{P}(R(i,j)=m)=\binom{M}{m}(\frac{\mathscr{R}(i,j)}{M})^{m}(1-\frac{\mathscr{R}(i,j)}{M})^{M-m} \qquad m=0,1,2,\ldots,M,
\end{align}
where $\binom{M}{m}$ represents the binomial coefficient. Recall that $\Theta$'s elements are required to be in the range $[0,M]$, this is because $\Pi$'s entries range in $[0,1]$, $\sum_{k=1}^{K}\Pi(i,k)=1$ for $i\in[N], \mathscr{R}=\Pi\Theta'$, and $\frac{\mathscr{R}(i,j)}{M}$ is a probability in $[0,1]$.
\begin{rem}
For the binary responses case (i.e., $R$'s elements are either 0 or 1), we can simply set $M$ as 1 in Equation (\ref{GoM}). For this case, Equation (\ref{GoMBinomial}) becomes $\mathbb{P}(R(i,j)=1)=\mathscr{R}(i,j)$ and $\mathbb{P}(R(i,j)=0)=1-\mathscr{R}(i,j)$.
\end{rem}

By Theorem 2 in \citep{chen2024spectral}, we know that the GoM model is identifiable (i.e., well-defined) when the following conditions hold
\begin{description}
  \item[(C1)] Each latent class has at least one pure subject.
  \item[(C2)] The rank of the $J\times K$ matrix $\Theta$ is $K$.
\end{description}

Throughout this paper, the two conditions (C1) and (C2) are treated as default. Let $\mathcal{I}$ be the index set of pure subjects such that $\mathcal{I}=\{p_{1}, p_{2}, \ldots,p_{K}\}$ with $p_{k}$ being an arbitrary pure subject in the $k$-th latent class for $k\in[K]$. Similar to \citep{mao2021estimating}, W.L.O.G., reorder the subjects so that $\Pi(\mathcal{I},:)=I_{K\times K}$, the $K$-by-$K$ identity matrix.

By Equation (\ref{GoM}), a response matrix $R\in\{0,1,2,\ldots,M\}^{N\times J}$ can be generated from the GoM mode for categorical data by below two steps.
\begin{description}
  \item[(a)] Set $\Pi$ and $\Theta$ satisfying Conditions (C1) and (C2). Set $\mathscr{R}=\Pi\Theta'$.
  \item[(b)] Generate $R(i,j)$ from a Binomial distribution with $M$ independent trails and success probability $\frac{\mathscr{R}(i,j)}{M}$ for $i\in[N], j\in[J]$.
\end{description}

By Steps (a) and (b), we can generate a response matrix $R$ with elements taking values in $\{0,1,\ldots, M\}$, true mixed membership matrix $\Pi$, and true item parameter matrix $\Theta$. A visualization of a response matrix $R$ as taking $\Pi$ and $\Theta$ as input in the GoM model when $K=2, N=20$, and $J=10$ can be found in Fig.~\ref{Ex4simR}. We see that such $R$ matches the real-world observed response matrix. Given $R$, the primal goal for latent class analysis is to estimate $\Pi$ and $\Theta$.
\section{Algorithms}\label{sec3}
To facilitate readers' understanding of our algorithms, we start from the oracle case where $R$'s expectation $\mathscr{R}$ is assumed to be known in advance. When Conditions (C1) and (C2) are satisfied, $\mathscr{R}$'s rank is $K$ since $\mathscr{R}=\Pi\Theta'$. Recall that $\mathscr{R}$ is an $N\times J$ matrix and $K\ll\mathrm{min}(N,J)$, as a result, $\mathscr{R}$ has only $K$ nonzero singular values. First, we define the following matrix based on $\mathscr{R}$,
\begin{align}\label{populationL}
\mathscr{L}_{\tau}=\mathscr{D}^{-1/2}_{\tau}\mathscr{R},
\end{align}
where $\mathscr{D}_{\tau}=\mathscr{D}+\tau I_{N\times N}$, $\mathscr{D}$ is an $N\times N$ diagonal matrix with $(i,i)$-th diagonal element being $\mathscr{D}(i,i)=\sum_{j=1}^{J}\mathscr{R}(i,j)$, and $\tau\geq0$ is the regularization parameter (we also call it regularizer occasionally). Sure, $\mathscr{D}$'s rank is $N$, which gives that the rank of $\mathscr{L}_{\tau}$ is $K$. Call $\mathscr{L}_{\tau}$ population regularized Laplacian matrix in this paper.

For any vector $x$, $\|x\|_{q}$ denotes its $l_{q}$-norm and $\mathrm{diag}(x)$ denotes a diagonal matrix with elements of $x$ on its diagonal. Set $X(:,j), X(S,:), \|X\|_{F}$, and $\sigma_{k}(X)$ as $X$'s $j$-th column, submatrix made by rows in set $S$, Frobenius norm, and the $k$-th largest singular value for any matrix $X$. Analyzing the singular value decomposition (SVD) of $\mathscr{L}_{\tau}$ gives the following lemma which serves as the starting point of our algorithms.
\begin{lem}\label{SVDPopulationLtau}
Let $\mathscr{L}_{\tau}=U\Sigma V'$ be $\mathscr{L}_{\tau}$'s top-$K$ SVD  such that $\Sigma=\mathrm{diag}(\sigma_{1}(\mathscr{L}_{\tau}), \sigma_{2}(\mathscr{L}_{\tau}),\ldots,\sigma_{K}(\mathscr{L}_{\tau}))$,  $U=[\eta_{1},\eta_{2},\ldots,\eta_{K}]$, and $V=[\xi_{1},\xi_{2},\ldots,\xi_{K}]$, where $U'U=I_{K\times K}, V'V=I_{K\times K}$, $\eta_{k}$ and $\xi_{k}$ denote the left and right singular vectors of $\sigma_{k}(\mathscr{L}_{\tau})$, respectively. Define two $N\times K$ matrices $U_{\tau}$ and $U_{*}$ as $U_{\tau}=\mathscr{D}^{1/2}_{\tau}U$ and $U_{*}=D_{U}U$, where $D_{U}$ is a diagonal matrix and $D_{U}(i,i)=\frac{1}{\|U(i,:)\|_{F}}$ for $i\in[N]$. Then, we have
\begin{itemize}
  \item ($\Theta$'s alternative form) $\Theta=\mathscr{R}'\Pi(\Pi'\Pi)^{-1}$.
  \item (Ideal Simplex) $U_{\tau}=\Pi U_{\tau}(\mathcal{I},:)$.
  \item (Ideal Cone) $U_{*}=YU_{*}(\mathcal{I},:)$, where $Y=D_{o}\Pi\mathscr{D}^{1/2}_{\tau}(\mathcal{I},\mathcal{I})D^{-1}_{U}(\mathcal{I},\mathcal{I})$ and $D_{o}$ is an $N\times N$ diagonal matrix with  $(i,i)$-th diagonal entry being $\frac{1}{\|\Pi(i,:)U_{\tau}(\mathcal{I},:)\|_{F}}$ for $i\in[N]$.
\end{itemize}
\end{lem}

\begin{figure}
\centering
\vspace{-0.25cm}
\subfigcapskip=-20pt
\subfigure[$U_{\tau}$]{\includegraphics[width=0.35\textwidth]{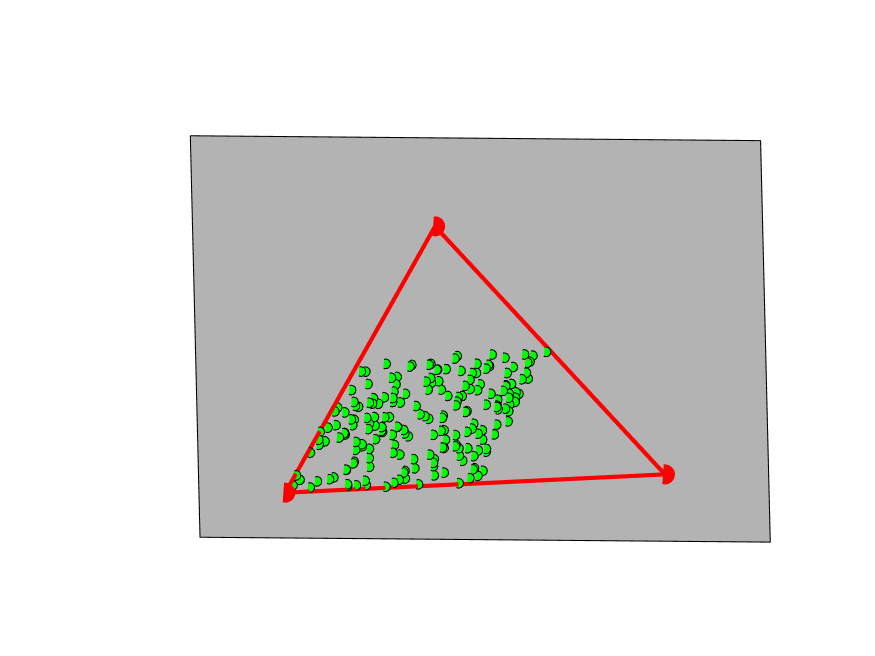}}
\subfigcapskip=-20pt
\subfigure[$U_{*}$]{\includegraphics[width=0.35\textwidth]{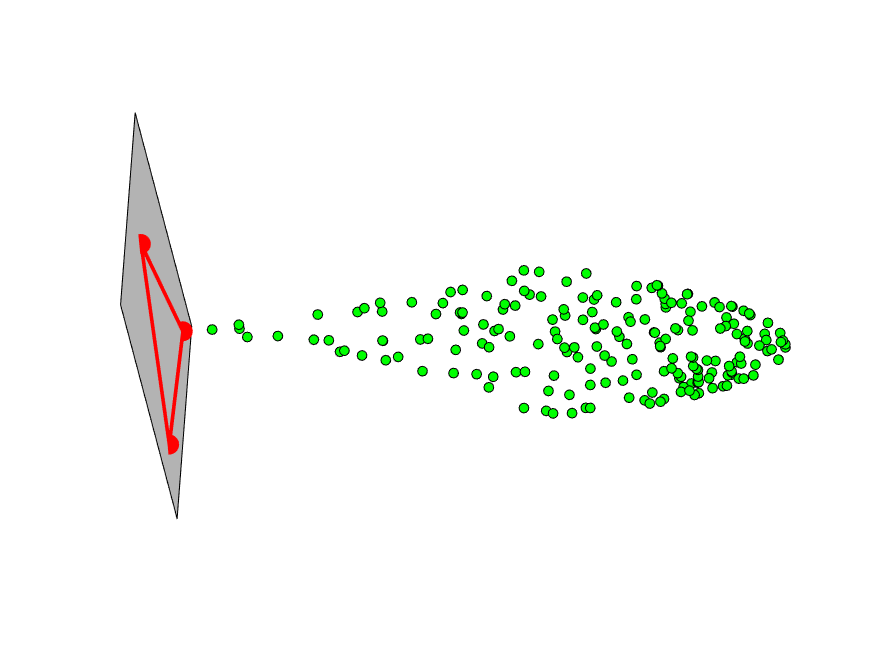}}
\vspace{-0.35cm}
\caption{Panel (a): illustration of the simplex structure of $U_{\tau}$ with $K=3$, where dots denote rows of $U_{\tau}$. Panel (b): illustration of the cone structure of $U_{*}$ with $K=3$, where dots denote rows of $U_{*}$. For both panels, red dots denote pure rows and green dots represent mixed rows. For both panels, the gray plane denotes the hyperplane formed by pure rows. For the simplex structure in panel (a), all mixed rows of $U_{\tau}$ (i.e., green dots) lie in the triangle formed by the three rows of $U_{\tau}(\mathcal{I},:)$. For the cone structure in panel (b), all mixed rows of $U_{*}$ (i.e., green dots) locate at one side of the hyperplane formed by the three rows of $U_{*}(\mathcal{I},:)$. For both panels, the settings of $\Pi$ and $\Theta$ are the same as that of Experiment 1 in Section \ref{SecSim}. Points in this figure have been projected and rotated from $\mathbb{R}^{3}$ to $\mathbb{R}^{2}$ for visualization.}
	\label{PlotIdealSC}
\vspace{-0.35cm}	
\end{figure}

Call the $i$-th row of $U_{\tau}$ (and $U_{*}$) a pure row if subject $i$ is pure and a mixed row otherwise. For the simplex structure of $U_{\tau}$, a pure row is one of the $K$ vertices of $U_{\tau}(\mathcal{I},:)$ and a mixed row lies inside of the simplex formed by $U_{\tau}(\mathcal{I},:)$'s $K$ rows. For the cone structure of $U_{*}$, all mixed rows locate at one side of the hyperplane formed by the $K$ rows of $U_{*}(\mathcal{I},:)$. Fig.~\ref{PlotIdealSC} gives illustrations of the simplex structure and the cone structure when $K=3$.

With the aid of the Ideal Simplex and Ideal Cone structures as outlined in Lemma \ref{SVDPopulationLtau}, we can devise two spectral methods. Subsequently, we delve into the development of our proposed methods, preceded by concise explanations of these two structures.
\subsection{The GoM-SRSC algorithm}
The form $U_{\tau}=\Pi U_{\tau}(\mathcal{I},:)$ is known as ideal simplex which comes from the fact that $U_{\tau}(i,:)$ is a convex linear combination of $U_{\tau}(\mathcal{I},:)$'s $K$ rows through $U_{\tau}=\Pi U_{\tau}(\mathcal{I},:)$ for $i\in[N]$. In detail, denote the simplex made by $U_{\tau}(\mathcal{I},:)$'s $K$ rows as $\mathcal{S}^{ideal}(s_{1},s_{2},\ldots,s_{K})$ with $s_{k}$ being $U_{\tau}(\mathcal{I},:)$'s $k$-th row for $k\in[K]$. We have $U_{\tau}(i,:)=\sum_{k=1}^{K}\Pi(i,k)s_{k}$, which suggests that $U_{\tau}(i,:)$ is one of the $K$ vertices of the simplex $\mathcal{S}^{ideal}(s_{1},s_{2},\ldots,s_{K})$ if subject $i$ is pure while $U_{\tau}(i,:)$ falls inside of the simplex otherwise because $\Pi\in[0,1]^{N\times K}, \|\Pi(i,:)\|_{1}=1$, and $\Pi$ satisfies Condition (C1). The ideal simplex structure is also found in the SVD of $\mathscr{R}$ under GoM in \citep{chen2024spectral} and the problem of estimating mixed memberships for social networks in \citep{mao2021estimating,qing2024bipartite}.

Since $U_{\tau}(\mathcal{I},:)$ is nonsingular, we have $\Pi=U_{\tau}U^{-1}_{\tau}(\mathcal{I},:)$, which suggests that as long as we know the index set $\mathcal{I}$, we can recover $\Pi$ using $U_{\tau}U^{-1}_{\tau}(\mathcal{I},:)$. As suggested by \citep{mao2018overlapping}, running the successive projection (SP) algorithm \citep{gillis2015semidefinite} to $U_{\tau}$'s rows can exactly return the index set $\mathcal{I}$, where SP's detail can be found in Algorithm 1 of \citep{gillis2015semidefinite}. Set $Z=U_{\tau}U^{-1}_{\tau}(\mathcal{I},:)\equiv \Pi$, we have $\Pi(i,:)=\frac{Z(i,:)}{\|Z(i,:)\|_{1}}$ for $i\in[N]$. After recovering $\Pi$, $\Theta$ can be recovered by the 1st statement of Lemma \ref{SVDPopulationLtau}. The above analysis suggests the following algorithm which we call Ideal GoM-SRSC, where SRSC is short for simplex regularized spectral clustering. Input: $\mathscr{R}, K$, and $\tau\geq0$. Output: $\Pi$ and $\Theta$.
\begin{itemize}
  \item Compute $\mathscr{L}_{\tau}$ using Equation (\ref{populationL}).
  \item Get $U\Sigma V'$, the compact SVD of $\mathscr{L}_{\tau}$. Set $U_{\tau}=\mathscr{D}^{1/2}_{\tau}U$.
  \item Apply the SP algorithm to $U_{\tau}$'s rows with $K$ vertices to obtain $U_{\tau}(\mathcal{I},:)$.
  \item Set $Z=U_{\tau}U^{-1}_{\tau}(\mathcal{I},:)$.
  \item Set $\Pi(i,:)=\frac{Z(i,:)}{\|Z(i,:)\|_{1}}$ for $i\in[N]$.
  \item Set $\Theta=\mathscr{R}'\Pi(\Pi'\Pi)^{-1}$.
\end{itemize}

In practice, we only observe the response matrix $R$ and we aim at providing good estimations of $\Pi$ and $\Theta$ from $R$. To extend the ideal algorithm to the real case, first, we introduce the regularized Laplacian matrix as below:
\begin{align}\label{realLtau}
L_{\tau}=D^{-1/2}_{\tau}R,
\end{align}
where $D_{\tau}=D+\tau I_{N\times N}$, and the $(i,i)$-th diagonal entry of the $N\times N$ diagonal matrix $D$ is $D(i,i)=\sum_{j=1}^{J}R(i,j)$ for $i\in[N]$. The regularized Laplacian matrix $L_{\tau}$ has been used in \citep{qing2023latent} to estimate a latent class model. By comparing Equations (\ref{populationL}) and (\ref{realLtau}), we observe that $\mathscr{L}_{\tau}$ has a similar form as $L_{\tau}$ and it is $L_{\tau}$'s population version. Recall that $\mathscr{L}_{\tau}$ is a rank-$K$ matrix, it is expected that the top-$K$ SVD of $L_{\tau}$  satisfactorily approximates $\mathscr{L}_{\tau}$. Therefore, we let $\hat{L}_{\tau}=\hat{U}\hat{\Sigma}\hat{V}'$ be the top-$K$ SVD of $L_{\tau}$ such that $\hat{\Sigma}=\mathrm{diag}(\sigma_{1}(L_{\tau}),\sigma_{2}(L_{\tau}),\ldots,\sigma_{K}(L_{\tau}))$, $\hat{U}=[\hat{\eta}_{1},\hat{\eta}_{2}, \ldots,\hat{\eta}_{K}], \hat{V}=[\hat{\xi}_{1}, \hat{\xi}_{2}, \ldots,\hat{\xi}_{K}]$, $\hat{U}$ and $\hat{V}$ satisfy $\hat{U}'\hat{U}=I_{K\times K}$ and $\hat{V}'\hat{V}=I_{K\times K}$, where $\hat{\eta}_{k}$ and $\hat{\xi}_{k}$ denote the left and right singular vector of the $k$-th largest singular value $\hat{\sigma}_{k}(L_{\tau})$, respectively. Let $\hat{U}_{\tau}=D^{1/2}_{\tau}\hat{U}$ be an approximation of $U_{\tau}$. Then applying the SP algorithm to $\hat{U}_{\tau}$'s rows with $K$ vertices gets the estimated index set $\hat{\mathcal{I}}$, which should be a good estimation of the index set $\mathcal{I}$ when the observed response matrix $R$ is generated from the GoM model with expectation $\mathscr{R}$. This is illustrated by the panel (a) of Fig.~\ref{PlotRealSC}, which demonstrates that the empirical simplex structure formed by the estimated pure nodes in $\hat{U}_{\tau}$ found by the SP algorithm closely resembles the Ideal Simplex in $U_{\tau}$. Let $\hat{Z}=\mathrm{max}(0,\hat{U}_{\tau}\hat{U}^{-1}_{\tau}(\hat{\mathcal{I}},:))$ be an estimation of $Z$, where we only consider the nonnegative part of $\hat{U}_{\tau}\hat{U}^{-1}_{\tau}(\hat{\mathcal{I}},:)$ since all entries of $Z$ are nonnegative. We estimate the membership score $\Pi(i,:)$ for the $i$-th subject by using $\hat{\Pi}=\frac{\hat{Z}(i,:)}{\|\hat{Z}(i,:)\|_{1}}$ for $i\in[N]$. Finally, we estimate the item parameter matrix $\Theta$ by using $\hat{\Theta}$ calculated as $\hat{\Theta}=\mathrm{min}(M,\mathrm{max}(0,R'\hat{\Pi}(\hat{\Pi}'\hat{\Pi})^{-1}))$, where we truncate all entries of $R'\hat{\Pi}(\hat{\Pi}'\hat{\Pi})^{-1}$ to be in the range $[0,M]$ since $\Theta$'s elements are in the range $[0,M]$ theoretically. We summarize the above analysis into the following algorithm which extends the Ideal GoM-SRSC algorithm naturally to the real case with known the observed response matrix $R$. The flowchart of GoM-SRSC is displayed in Fig.~\ref{IllusAlgorithm1}.
\begin{figure}
\centering
\vspace{-0.25cm}
\subfigcapskip=-20pt
\subfigure[]{\includegraphics[width=0.35\textwidth]{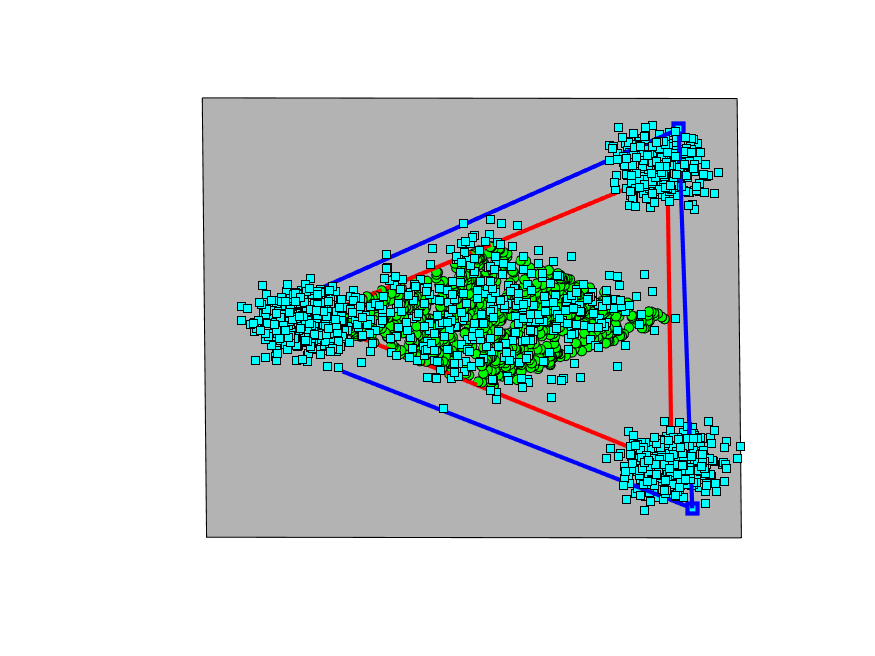}}
\subfigcapskip=-20pt
\subfigure[]{\includegraphics[width=0.35\textwidth]{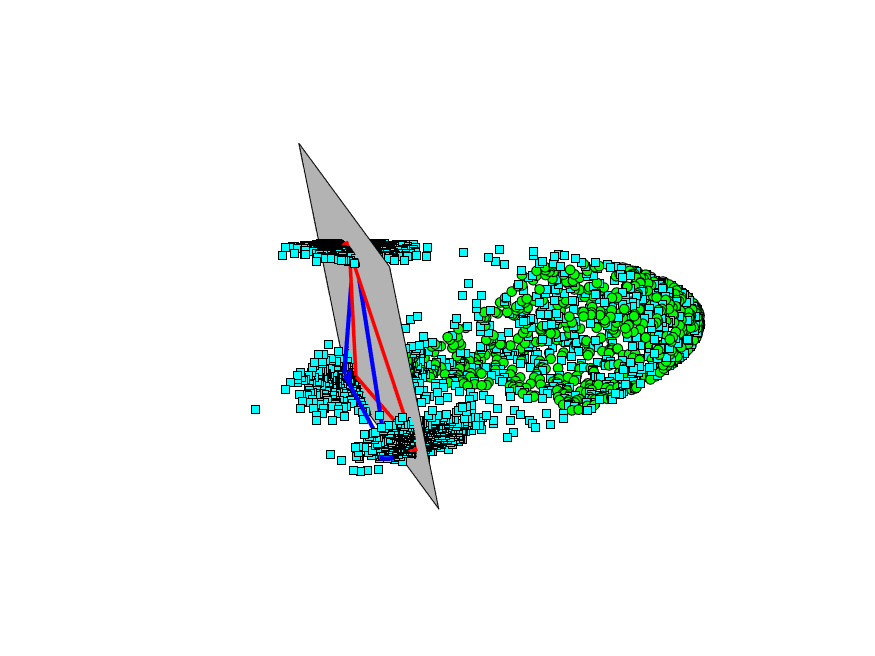}}
\vspace{-0.35cm}
\caption{Panel (a): row vectors of $U_{\tau}$ and $\hat{U}_{\tau}$ projected from $\mathbb{R}^{3}$ to $\mathbb{R}^{2}$. Panel (b): row vectors of $U_{*}$ and $\hat{U}_{*}$. For both panels, red dots (vertexes of the two red triangles) represent pure nodes, green dots represent mixed nodes, and the two gray planes denote the hyperplanes formed by pure rows. For panel (a), cyan squares represent rows of $\hat{U}_{\tau}$ and blue squares (vertexes of the blue triangle in panel (a)) represent the estimated pure nodes in $\hat{U}_{\tau}$ found by the SP algorithm. For panel (b), cyan squares represent rows of $\hat{U}_{*}$ and blue squares (vertexes of the blue triangle in panel (b)) represent the estimated pure nodes in $\hat{U}_{*}$ found by the SVM-cone algorithm. For both panels, the settings of $\Pi$ and $\Theta$ are the same as those of Experiment 3 in Section \ref{SecSim} when $N=3200, J=800, K=3,$ and $\rho=1$.}
\label{PlotRealSC}
\vspace{-0.35cm}	
\end{figure}

\begin{algorithm}
\caption{\textbf{GoM-SRSC}}
\label{alg:SRSC}
\begin{algorithmic}[1]
\Require The observed response matrix $R\in\{0,1,2,\ldots, M\}^{N\times J}$, the number of latent classes $K$, and the regularization parameter $\tau$  (a good choice for $\tau$ is $M\mathrm{max}(N,J)$.).
\Ensure The estimated membership matrix $\hat{\Pi}$ and the estimated item parameter matrix $\hat{\Theta}$.
\State Get $L_{\tau}$ by Equation (\ref{realLtau}).
\State Get $\hat{U}\hat{\Sigma}\hat{V}'$, the top-$K$ singular value decomposition of $L_{\tau}$. Set $\hat{U}_{\tau}=D^{1/2}_{\tau}\hat{U}$.
\State Apply the SP algorithm on $\hat{U}_{\tau}$'s rows with $K$ vertices to obtain $\hat{U}_{\tau}(\hat{\mathcal{I}},:)$, where $\hat{\mathcal{I}}$ represents the estimated index set returned by SP.
\State Set $\hat{Z}=\mathrm{max}(0,\hat{U}_{\tau}\hat{U}^{-1}_{\tau}(\hat{\mathcal{I}},:))$.
\State Set $\hat{\Pi}(i,:)=\frac{\hat{Z}(i,:)}{\|\hat{Z}(i,:)\|_{1}}$ for $i\in[N]$.
\State Set $\hat{\Theta}=\mathrm{min}(M,\mathrm{max}(0,R'\hat{\Pi}(\hat{\Pi}'\hat{\Pi})^{-1}))$
\end{algorithmic}
\end{algorithm}

The computational cost of our GoM-SRSC algorithm mainly comes from the top $K$ SVD of $L_{\tau}$, the SP algorithm, and the last step in estimating $\Theta$. The complexities of these three main steps are $O(\mathrm{max}(N^{2},J^{2})K), O(NK^{2})$, and $O(NJK)$, respectively. Because $K\ll \mathrm{min}(N,J)$ in this paper, as a result, our GoM-SRSC algorithm has a total complexity $O(\mathrm{max}(N^{2},J^{2})K)$.
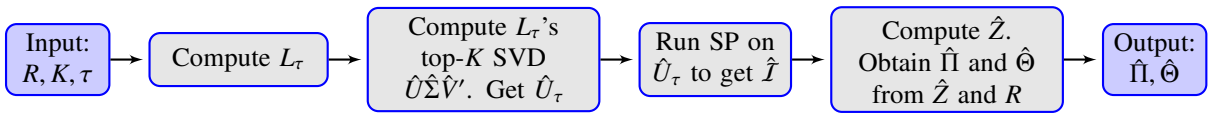
\begin{figure}[H]
\centering
\begin{tikzpicture}
[auto,
decision/.style={diamond, draw=blue, thick, fill=blue!20,
    text width=4.5em,align=flush center,
    inner sep=1pt},
block/.style ={rectangle, draw=blue, thick, fill=gray!20,
    text width=8em,align=center, rounded corners,
    minimum height=2em},
block1/.style ={diamond, draw=blue, thick, fill=orange!20,
    text width=1.4em,align=center, rounded corners,
    minimum height=2em},
block2/.style ={rectangle, draw=blue, thick, fill=blue!20,
    text width=3.2em,align=center, rounded corners,
    minimum height=2em},
block33/.style ={rectangle, draw=blue, thick, fill=gray!20,
    text width=6em,align=center, rounded corners,
    minimum height=2em},
block3/.style ={rectangle, draw=blue, thick, fill=gray!20,
    text width=8em,align=center, rounded corners,
    minimum height=2em},
block4/.style ={rectangle, draw=blue, thick, fill=gray!20,
    text width=5em,align=center, rounded corners,
    minimum height=2em},
line/.style ={draw, thick, -latex',shorten >=2pt},
cloud/.style ={draw=red, thick, ellipse,fill=red!20,
    minimum height=2em}]
\matrix [column sep=5mm,row sep=7mm]
{
&\node [block2] (input) {Input: $R, K, \tau$};
&\node [block33] (pca) {Compute $L_{\tau}$};
&\node [block3] (pca1) {Compute $L_{\tau}$'s top-$K$ SVD $\hat{U}\hat{\Sigma}\hat{V}'$. Get $\hat{U}_{\tau}$};
&\node [block4] (vh) {Run SP on $\hat{U}_{\tau}$ to get  $\hat{\mathcal{I}}$};
&\node [block] (mr) {Compute $\hat{Z}$. Obtain $\hat{\Pi}$ and $\hat{\Theta}$ from $\hat{Z}$ and $R$};
&\node [block2] (output) {Output: $\hat{\Pi}, \hat{\Theta}$};\\
};
\begin{scope}[every path/.style=line]
\path (input) -- (pca);
\path (pca) -- (pca1);
\path (pca1) -- (vh);
\path (vh) -- (mr);
\path (mr) -- (output);
\end{scope}
\end{tikzpicture}
\caption{Flowchart of Algorithm \ref{alg:SRSC}.}\label{IllusAlgorithm1}
\end{figure}
\subsection{The GoM-CRSC algorithm}
$U_{*}=YU_{*}(\mathcal{I},:)$ forms a cone structure introduced in Problem 1 of \citep{mao2018overlapping} since the $l_{2}$ norm of  $U_{*}(i,:)$ is 1 for $i\in[N]$. This form is different from the ideal simplex structure because the $N\times K$ matrix $Y$ has a different property as the membership matrix $\Pi$. Unlike the ideal simplex structure, applying the successive projection algorithm to all rows of $U_{*}$ cannot recover the $K\times K$ matrix $U_{*}(\mathcal{I},:)$. Luckily, as analyzed in \citep{mao2018overlapping}, running the SVM-cone algorithm of \citep{mao2018overlapping} to all rows of $U_{*}$ can exactly recover $U_{*}(\mathcal{I},:)$ as long as $(U_{*}(\mathcal{I},:)U'_{*}(\mathcal{I},:))^{-1}\mathbf{1}>0$ holds (here, $\mathbf{1}$ is a vector with all entries being 1) which is guaranteed by Lemma \ref{ConeCon} below, where the SVM-cone algorithm is summarized in Algorithm 1 of \citep{mao2018overlapping}.
\begin{lem}\label{ConeCon}
Under $\mathrm{GoM}(\Pi,\Theta)$, $(U_{*}(\mathcal{I},:)U'_{*}(\mathcal{I},:))^{-1}\mathbf{1}>0$ holds.
\end{lem}
Therefore, we can run the SVM-cone algorithm to $U_{*}$'s rows to recover $U_{*}(\mathcal{I},:)$. Then we can recover $Y$ by setting $Y=U_{*}U^{-1}_{*}(\mathcal{I},:)$. Set $Z_{*}=UU^{-1}_{*}(\mathcal{I},:)D_{U}(\mathcal{I},\mathcal{I})\mathscr{D}^{-1/2}_{\tau}(\mathcal{I},\mathcal{I})$. The following lemma guarantees that we can recover $\Pi$ from $Z_{*}$.
\begin{lem}\label{ZstarPi}
Under $\mathrm{GoM}(\Pi,\Theta)$, we have $\Pi(i,:)=\frac{Z_{*}(i,:)}{\|Z_{*}(i,:)\|_{1}}$ for $i\in[N]$.
\end{lem}

After getting $\Pi$ from $Z_{*}$ by Lemma \ref{ZstarPi}, $\Theta$ can be obtained by the 1st result in Lemma \ref{SVDPopulationLtau}. Then we get the Ideal GoM-CRSC algorithm below, where CRSC means cone regularized spectral clustering. Input: $\mathscr{R}, K$, and $\tau\geq0$. Output: $\Pi$ and $\Theta$.
\begin{itemize}
  \item Compute $\mathscr{L}_{\tau}$ using Equation (\ref{populationL}).
  \item Get $U\Sigma V'$, the top-$K$ SVD of $\mathscr{L}_{\tau}$. Obtain $U_{*}$ from $U$.
  \item Run the SVM-cone algorithm with inputs $U_{*}$ and $K$ to obtain $U_{*}(\mathcal{I},:)$.
  \item Set $Z_{*}=UU^{-1}_{*}(\mathcal{I},:)D_{U}(\mathcal{I},\mathcal{I})\mathscr{D}^{-1/2}_{\tau}(\mathcal{I},\mathcal{I})$.
  \item Set $\Pi(i,:)=\frac{Z_{*}(i,:)}{\|Z_{*}(i,:)\|_{1}}$ for $i\in[N]$.
  \item Set $\Theta=\mathscr{R}'\Pi(\Pi'\Pi)^{-1}$.
\end{itemize}

Now, we consider the real case with known $R$ instead of its expectation $\mathscr{R}$. Let $\hat{U}_{*}=D_{\hat{U}}\hat{U}$, where $D_{\hat{U}}$ is an $N\times N$ diagonal matrix with $(i,i)$-th diagonal entry being $\frac{1}{\|\hat{U}(i,:)\|_{F}}$ for $i\in[N]$. Panel (b) of Fig.~\ref{PlotRealSC} illustrates that the empirical cone structure formed by the estimated pure nodes in $\hat{U}_{*}$ found by the SVM-cone algorithm closely resembles the Ideal Cone in $U_{*}$. The following algorithm extends the Ideal GoM-CRSC algorithm to the real case naturally. The flowchart of GoM-CRSC is displayed in Fig.~\ref{IllusAlgorithm2}. The computational cost of SVM-cone is $O(KN^{2})$ \citep{chang2011libsvm,mao2018overlapping}. Combing the complexity analysis of Algorithm \ref{alg:SRSC}, GoM-CRSC's complexity is $O(\mathrm{max}(N^{2},J^{2})K)$.
\begin{algorithm}
\caption{\textbf{GoM-CRSC}}
\label{alg:CRSC}
\begin{algorithmic}[1]
\Require $R, K$, and $\tau$  (a good default choice for $\tau$ is $M\mathrm{max}(N,J)$.).
\Ensure $\hat{\Pi}_{*}$ and $\hat{\Theta}_{*}$.
\State Get $L_{\tau}$ by Equation (\ref{realLtau}).
\State Get $\hat{U}\hat{\Sigma}\hat{V}'$, the top-$K$ SVD of $L_{\tau}$. Set $\hat{U}_{*}=D_{\hat{U}}\hat{U}$.
\State Run the SVM-cone algorithm with inputs $\hat{U}_{*}$ and $K$ to obtain $\hat{U}_{*}(\hat{\mathcal{I}}_{*},:)$, where $\hat{\mathcal{I}}_{*}$ is the estimated index set returned by SVM-cone.
\State Set $\hat{Z}_{*}=\mathrm{max}(0,\hat{U}\hat{U}^{-1}_{*}(\hat{\mathcal{I}}_{*},:)D_{\hat{U}}(\hat{\mathcal{I}}_{*},\hat{\mathcal{I}}_{*})D^{-1/2}_{\tau}(\hat{\mathcal{I}}_{*},\hat{\mathcal{I}}_{*}))$.
\State Set $\hat{\Pi}_{*}(i,:)=\frac{\hat{Z}_{*}(i,:)}{\|\hat{Z}_{*}(i,:)\|_{1}}$ for $i\in[N]$.
\State Set $\hat{\Theta}_{*}=\mathrm{min}(M,\mathrm{max}(0,R'\hat{\Pi}_{*}(\hat{\Pi}'_{*}\hat{\Pi}_{*})^{-1}))$
\end{algorithmic}
\end{algorithm}
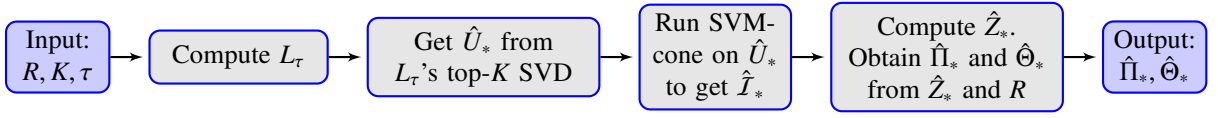
\begin{figure}[H]
\centering
\begin{tikzpicture}
[auto,
decision/.style={diamond, draw=blue, thick, fill=blue!20,
    text width=4.5em,align=flush center,
    inner sep=1pt},
block/.style ={rectangle, draw=blue, thick, fill=gray!20,
    text width=8em,align=center, rounded corners,
    minimum height=2em},
block1/.style ={diamond, draw=blue, thick, fill=orange!20,
    text width=1.4em,align=center, rounded corners,
    minimum height=2em},
block2/.style ={rectangle, draw=blue, thick, fill=blue!20,
    text width=3.2em,align=center, rounded corners,
    minimum height=2em},
block33/.style ={rectangle, draw=blue, thick, fill=gray!20,
    text width=6em,align=center, rounded corners,
    minimum height=2em},
block3/.style ={rectangle, draw=blue, thick, fill=gray!20,
    text width=8em,align=center, rounded corners,
    minimum height=2em},
block4/.style ={rectangle, draw=blue, thick, fill=gray!20,
    text width=5em,align=center, rounded corners,
    minimum height=2em},
line/.style ={draw, thick, -latex',shorten >=2pt},
cloud/.style ={draw=red, thick, ellipse,fill=red!20,
    minimum height=2em}]
\matrix [column sep=5mm,row sep=7mm]
{
&\node [block2] (input) {Input: $R, K, \tau$};
&\node [block33] (pca) {Compute $L_{\tau}$};
&\node [block3] (pca1) {Get $\hat{U}_{*}$ from $L_{\tau}$'s top-$K$ SVD};
&\node [block4] (vh) {Run SVM-cone on $\hat{U}_{*}$ to get  $\hat{\mathcal{I}}_{*}$};
&\node [block] (mr) {Compute $\hat{Z}_{*}$. Obtain $\hat{\Pi}_{*}$ and $\hat{\Theta}_{*}$ from $\hat{Z}_{*}$ and $R$};
&\node [block2] (output) {Output: $\hat{\Pi}_{*}, \hat{\Theta}_{*}$};\\
};
\begin{scope}[every path/.style=line]
\path (input) -- (pca);
\path (pca) -- (pca1);
\path (pca1) -- (vh);
\path (vh) -- (mr);
\path (mr) -- (output);
\end{scope}
\end{tikzpicture}
\caption{Flowchart of Algorithm \ref{alg:CRSC}.}\label{IllusAlgorithm2}
\end{figure}
\section{Asymptotic consistency}\label{sec4}
Before providing the theoretical error rates of our methods, we introduce the sparsity parameter of an observed response matrix $R$. Recall that when $R\in\{0,1,2,\ldots,M\}^{N\times J}$, GoM requires $\Theta\in[0,M]^{J\times K}$, which implies that the maximum entry of $\Theta$ should be no larger than $M$. Let $\rho=\mathrm{max}_{j\in[J], k\in[K]}\Theta(j,k)$ and $B=\frac{\Theta}{\rho}$, we see that $\rho$ ranges in $[0,M]$ and $\mathrm{max}_{j\in[J], k\in[K]}B(j,k)=1$. By Equation (\ref{GoMBinomial}) and $\mathscr{R}(i,j)=\Pi(i,:)\Theta'(j,:)=\rho \Pi(i,:)B'(j,:)$, we see that the probability of no-response (i.e., $R(i,j)=0$) equals $(1-\frac{\rho\Pi(i,:)B'(j,:)}{M})^{M}$, a value decreases as $\rho$ increases, i.e., $\rho$ controls the sparsity (i.e., number of zeros) of an observed response matrix $R$. For this reason, we call $\rho$ the sparsity parameter. In real-world categorical data such as psychological test datasets and educational assessment datasets, there usually exist several subjects that do not respond to all items, i.e., there exist some no-responses for real-world categorical data. Generally speaking, the task of estimating $\Pi$ and $\Theta$ is hard when categorical data is sparse. Therefore, it is critical to investigate the performances of different methods when categorical data has different levels of sparsity. We will study $\rho$'s influence on the performance of our approaches both theoretically and numerically. Our results rely on the following sparsity assumption, which means a lower bound requirement on the sparsity parameter $\rho$ for our theoretical analysis.
\begin{assum}\label{Assum1}
 $\rho\mathrm{max}(N,J)\geq M^{2}\mathrm{log}(N+J)$.
\end{assum}
Assumption \ref{Assum1} is mild since it requires $\rho\geq\frac{M^{2}\mathrm{log}(N+J)}{\mathrm{max}(N,J)}$, where the lower bound is a small value when $N$ and $J$ are large. Meanwhile, it is reasonable to consider a lower bound requirement on $\rho$ since if $\rho$ is too small (i.e., there are too many zeros in $R$), it is impossible for any algorithm to have a satisfactory performance in estimating $\Pi$ and $\Theta$. We also need the following conditions to simplify our analysis.
\begin{con}\label{Con1}
$K=O(1), J=O(N), \sigma_{K}(\Pi)=O(\sqrt{\frac{N}{K}}), \sigma_{K}(B)=O(\sqrt{\frac{J}{K}})$, and $\pi_{\mathrm{min}}=O(\frac{N}{K})$, where $\pi_{\mathrm{min}}=\mathrm{min}_{k\in[K]}\sum_{i=1}^{N}\Pi(i,k)$.
\end{con}
Condition \ref{Con1} is also mild and this can be understood as below: $K=O(1)$ implies that $K$ is a constant number, $J=O(N)$ means that $J$ has the same order as $N$, $\sigma_{K}(\Pi)=O(\sqrt{\frac{N}{K}})$ and $\pi_{\mathrm{min}}=O(\frac{N}{K})$ means that each latent class has close size, and $\sigma_{K}(B)=O(\sqrt{\frac{J}{K}})$ means that the summation of each column of $B$ is in the same order. Let $e_{i}$ be a vector whose $i$-th element equals 1 and all the other entries are 0 in this paper. Theorem \ref{Main} below establishes the per-subject error rate in estimating $\Pi$ and the relative $l_{2}$ error in estimating $\Theta$ of our GoM-SRSC, where the error rates of our GoM-CRSC are the same as that of GoM-SRSC.
\begin{thm}\label{Main}
Under $\mathrm{GoM}(\Pi,\Theta)$, suppose Assumption \ref{Assum1} and Condition \ref{Con1} hold, with high probability, we have
\begin{align*}
\mathrm{max}_{i\in[N]}\|e'_{i}(\hat{\Pi}-\Pi\mathcal{P})\|_{1}=O(\sqrt{\frac{\mathrm{log}(N)}{\rho N}})\mathrm{~and~}\frac{\|\hat{\Theta}-\Theta\mathcal{P}\|_{F}}{\|\Theta\|_{F}}=O(\sqrt{\frac{\mathrm{log}(N)}{\rho N}}),
\end{align*}
where $\mathcal{P}$ is a $K\times K$ permutation matrix.
\end{thm}
We observe that the error rates decrease as the sparsity parameter $\rho$ increases according to Theorem \ref{Main}, which is in line with our intuition. We also observe that the error rates decrease to zero as the number of subjects $N$ and the number of items $J$ increase to infinity, which implies the estimation consistencies of our methods. Meanwhile, by the proof of Theorem \ref{Main}, $\tau\geq M\mathrm\mathrm{max}(N,J)$ should hold, where this requirement on $\tau$ is mainly used to simplify our theoretical bounds. \deleted{In fact, we} \added{We} find empirically in Section \ref{SecSim} that our algorithms GoM-SRSC and GoM-CRSC are insensitive to the choice of $\tau$.
\section{Quantifying the quality of latent mixed membership analysis}\label{sec5}
For real data, the true membership matrix $\Pi$ is often unknown, while an estimated membership matrix can always be obtained by applying a latent mixed membership analysis method to the observed response matrix $R$. Assuming $\hat{\Pi}_{1}$ and $\hat{\Pi}_{2}$ are the estimated membership matrices returned by two different algorithms applied to $R$ with the same number of latent classes. Since the real membership matrix $\Pi$ is unknown, it is natural to ask which algorithm returns a better partition for all subjects. This question prompts the need for a metric to assess the quality of the latent mixed membership analysis. Additionally, the number of latent classes, $K$, is usually not known for real-world categorical data, necessitating the development of a method to estimate it. In this paper, we address these two questions by providing a metric to measure the quality of the latent mixed membership analysis.

Next, we introduce our metric. Set $A=RR'$ and call $A$ adjacency matrix. We see that $A(i,\bar{i})=R(i,:)R'(\bar{i},:)=\sum_{j=1}^{J}R(i,j)R(\bar{i},j)$ for two distinct subjects $i$ and $\bar{i}$. To gain a better understanding of $A(i,\bar{i})$, we consider the binary response case. In this case, we see that $R(i,\bar{i})$ represents the common responses of $i$ and $\bar{i}$. Since subjects with similar memberships have similar response patterns, subjects with similar memberships should have more common responses than subjects with different memberships. This implies that $A$ can be viewed as the adjacency matrix of an assortative network in which nodes with similar memberships have more connections than nodes with different memberships \citep{newman2002assortative,newman2003mixing}. Our analysis suggests that applying the fuzzy modularity introduced in equation (14) of \citep{nepusz2008fuzzy} to $A$ is a good choice to measure the quality of the latent mixed membership analysis. Suppose that there is an $N\times k$ estimated membership matrix $\hat{\Pi}_{k}$ returned by applying algorithm $\mathcal{M}$ to $R$ with $k$ latent classes, where $\hat{\Pi}_{k}(i,:)\geq0$ and $\|\hat{\Pi}_{k}(i,:)\|_{1}=1$ for $i\in[N]$. The fuzzy modularity is computed in the following way:
\begin{align}\label{fuzzyModularity}
Q_{\mathcal{M}}(k)=\frac{1}{\omega}\sum_{i\in[N]}\sum_{\bar{i}\in[N]}(A(i,\bar{i})-\frac{d_{i}d_{\bar{i}}}{\omega})\hat{\Pi}_{k}(i,:)\hat{\Pi}'_{k}(\bar{i},:),
\end{align}
where $d_{i}=\sum_{j\in[N]}A(i,j)$ for $i\in[N], \omega=\sum_{i\in[N]}d_{i}$, and we denote the fuzzy modularity as $Q_{\mathcal{M}}(k)$ since it is computed using $\hat{\Pi}$, the estimated membership matrix of algorithm $\mathcal{M}$ with $k$ latent classes. When all subjects are pure, the fuzzy modularity reduces to the well-known Newman-Girvan modularity \citep{newman2004finding,newman2006modularity}. For the problem of latent mixed membership analysis for real-world categorical data, larger fuzzy modularity indicates better partition\added{,} and algorithms returning larger fuzzy modularity are preferred.

After defining the fuzzy modularity, we utilize the strategy proposed in \citep{nepusz2008fuzzy} to determine the optimal choice of $K$ for real categorical data. Specifically, we estimate $K$ by selecting the one that maximizes the fuzzy modularity. The pipeline of determining $K$ is shown in Fig.~\ref{IllusFindingK}. For convenience, when we apply the latent mixed membership analysis algorithm $\mathcal{M}$ on Equation (\ref{fuzzyModularity}) to estimate $K$, we refer to this method as K$\mathcal{M}$.
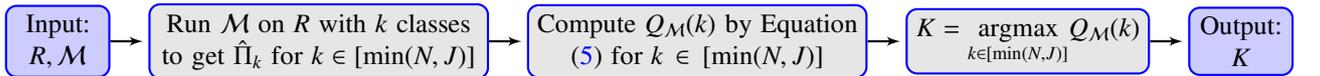
\begin{figure}[H]
\centering
\begin{tikzpicture}
[auto,
decision/.style={diamond, draw=blue, thick, fill=blue!20,
    text width=4.5em,align=flush center,
    inner sep=1pt},
block/.style ={rectangle, draw=blue, thick, fill=gray!20,
    text width=8.4em,align=center, rounded corners,
    minimum height=2em},
block1/.style ={diamond, draw=blue, thick, fill=orange!20,
    text width=1.4em,align=center, rounded corners,
    minimum height=2em},
block2/.style ={rectangle, draw=blue, thick, fill=blue!20,
    text width=3.2em,align=center, rounded corners,
    minimum height=2em},
block33/.style ={rectangle, draw=blue, thick, fill=gray!20,
    text width=6em,align=center, rounded corners,
    minimum height=2em},
block3/.style ={rectangle, draw=blue, thick, fill=gray!20,
    text width=8em,align=center, rounded corners,
    minimum height=2em},
block4/.style ={rectangle, draw=blue, thick, fill=gray!20,
    text width=12em,align=center, rounded corners,
    minimum height=2em},
line/.style ={draw, thick, -latex',shorten >=2pt},
cloud/.style ={draw=red, thick, ellipse,fill=red!20,
    minimum height=2em}]
\matrix [column sep=5mm,row sep=7mm]
{
&\node [block2] (input) {Input: $R, \mathcal{M}$};
&\node [block4] (vh) {Run $\mathcal{M}$ on $R$ with $k$ classes to get $\hat{\Pi}_{k}$ for $k\in[\mathrm{min}(N,J)]$};
&\node [block4] (mr) {Compute $Q_{\mathcal{M}}(k)$ by Equation (\ref{fuzzyModularity}) for $k\in[\mathrm{min}(N,J)]$};
&\node [block] (mr1) {$K=\underset{k\in[\mathrm{min}(N,J)]}{\mathrm{argmax}}Q_{\mathcal{M}}(k)$};
&\node [block2] (output) {Output: $K$};\\
};
\begin{scope}[every path/.style=line]
\path (input) -- (vh);
\path (vh) -- (mr);
\path (mr) -- (mr1);
\path (mr1) -- (output);
\end{scope}
\end{tikzpicture}
\caption{Pipeline of estimating $K$ for observed response matrix $R$ by combing the fuzzy modularity and algorithm $\mathcal{M}$.}\label{IllusFindingK}
\end{figure}
\section{Evaluation on synthetic categorical data}\label{sec6}
In this section, we present simulation studies. Firstly, we propose two alternative algorithms for GoM. Secondly, we introduce three evaluation metrics. Finally, we conduct thorough experimental studies.
\subsection{Two alternative algorithms for fitting GoM}
\subsubsection{GoM-SSC}
We call the first alternative method GoM-SSC, where SSC is short for simplex spectral clustering. Recall that $\mathscr{R}$'s rank is $K$ under GoM, without confusion, we let $\mathscr{R}=U\Sigma V'$ be $\mathscr{R}$'s compact SVD such that $U'U=I_{K\times K}$ and $V'V=I_{K\times K}$. Since $\mathscr{R}=\Pi\Theta'$, we have $U=\Pi\Theta'V\Sigma^{-1}=\Pi U(\mathcal{I},:)$, where $U=\Pi U(\mathcal{I},:)$ also forms a simplex structure. Then applying the SP algorithm to $U$ returns $U(\mathcal{I},:)$. Set $Z=UU^{-1}(\mathcal{I},:)$, we have $Z\equiv \Pi$ since $U=\Pi U(\mathcal{I},:)$. The above analysis suggests the following ideal algorithm named Ideal GoM-SSC. Input: $\mathscr{R}$ and $K$. Output: $\Pi$ and $\Theta$.
\begin{itemize}
  \item Let $U\Sigma V'$ be $\mathscr{R}$'s top-$K$ SVD.
  \item Apply the SP algorithm to $U$'s rows with $K$ vertices to get $U(\mathcal{I},:)$.
  \item Set $Z=UU^{-1}(\mathcal{I},:)$.
  \item Recover $\Pi$ and $\Theta$ by the last two steps of Ideal GoM-SRSC.
\end{itemize}

Algorithm \ref{alg:SSC} below extends the Ideal GoM-SSC to the real case with known $R$ naturally. Fig.~\ref{IllusAlgorithm3} displays the flowchart of GoM-SSC.
\begin{algorithm}
\caption{\textbf{GoM-SSC}}
\label{alg:SSC}
\begin{algorithmic}[1]
\Require $R$ and $K$.
\Ensure $\hat{\Pi}$ and $\hat{\Theta}$.
\State Get $\hat{U}\hat{\Sigma}\hat{V}'$, the top-$K$ SVD of $R$.
\State Apply the SP algorithm on $\hat{U}$'s rows with $K$ vertices to get $\hat{U}(\hat{\mathcal{I}},:)$.
\State $\hat{Z}_{}=\mathrm{max}(0,\hat{U}\hat{U}^{-1}(\hat{\mathcal{I}},:))$.
\State Obtain $\hat{\Pi}$ and $\hat{\Theta}$ by steps 5 and 6 of Algorithm \ref{alg:SRSC}.
\end{algorithmic}
\end{algorithm}
The computational cost of GoM-SSC is the same as our GoM-SRSC. Following a similar theoretical analysis to that of GoM-SRSC, we find that GoM-SSC's theoretical error rates are the same as that of GoM-SRSC under Assumption \ref{Assum1} and Condition \ref{Con1}. Note that GoM-SSC differs slightly from Algorithm 2 of \citep{chen2024spectral}. The slight differences are two-fold. First, similar to \citep{mao2021estimating}, Algorithm 2 of \citep{chen2024spectral} considers an extra pruning step to reduce noise. Second, $\hat{\Theta}$ in Algorithm 2 of \citep{chen2024spectral} is obtained by setting $\hat{\Theta}=\mathrm{min}(1-\epsilon, \mathrm{max}(\epsilon,\hat{V}\hat{\Sigma}\hat{U}'\hat{\Pi}(\hat{\Pi}'\hat{\Pi})^{-1}))$ with $\epsilon=0.001$ while $\hat{\Theta}$ is $\mathrm{min}(M,\mathrm{max}(0,R'\hat{\Pi}(\hat{\Pi}'\hat{\Pi})^{-1}))$ in our GoM-SSC and this difference occurs because we consider data with polytomous responses in this paper while \citep{chen2024spectral} only considers data with binary responses. As a result, GoM-SSC is a refinement of Algorithm 2 in \citep{chen2024spectral}.

\begin{figure}[H]
\centering
\begin{tikzpicture}
[auto,
decision/.style={diamond, draw=blue, thick, fill=blue!20,
    text width=4.5em,align=flush center,
    inner sep=1pt},
block/.style ={rectangle, draw=blue, thick, fill=gray!20,
    text width=8em,align=center, rounded corners,
    minimum height=2em},
block1/.style ={diamond, draw=blue, thick, fill=orange!20,
    text width=1.4em,align=center, rounded corners,
    minimum height=2em},
block2/.style ={rectangle, draw=blue, thick, fill=blue!20,
    text width=3.2em,align=center, rounded corners,
    minimum height=2em},
block33/.style ={rectangle, draw=blue, thick, fill=gray!20,
    text width=6em,align=center, rounded corners,
    minimum height=2em},
block3/.style ={rectangle, draw=blue, thick, fill=gray!20,
    text width=8em,align=center, rounded corners,
    minimum height=2em},
block4/.style ={rectangle, draw=blue, thick, fill=gray!20,
    text width=5em,align=center, rounded corners,
    minimum height=2em},
line/.style ={draw, thick, -latex',shorten >=2pt},
cloud/.style ={draw=red, thick, ellipse,fill=red!20,
    minimum height=2em}]
\matrix [column sep=5mm,row sep=7mm]
{
&\node [block2] (input) {Input: $R, K$};
&\node [block3] (pca) {Get $R$'s top-$K$ SVD $\hat{U}\hat{\Sigma}\hat{V}'$};
&\node [block4] (vh) {Run SP on $\hat{U}$ to get  $\hat{\mathcal{I}}$};
&\node [block] (mr) {Compute $\hat{Z}$. Obtain $\hat{\Pi}$ and $\hat{\Theta}$ from $\hat{Z}$ and $R$};
&\node [block2] (output) {Output: $\hat{\Pi}, \hat{\Theta}$};\\
};
\begin{scope}[every path/.style=line]
\path (input) -- (pca);
\path (pca) -- (vh);
\path (vh) -- (mr);
\path (mr) -- (output);
\end{scope}
\end{tikzpicture}
\caption{Flowchart of Algorithm \ref{alg:SSC}.}\label{IllusAlgorithm3}
\end{figure}
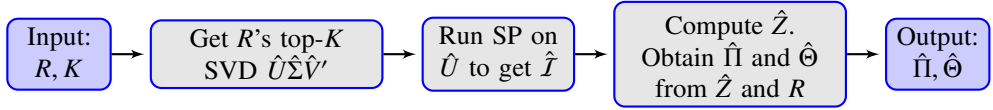
\subsubsection{GoM-SRM}
We call the second alternative method GoM-SRM, where SRM means simplex response matrix. Recall that $\mathscr{R}=\Pi\Theta'$ under GoM, we observe that $\mathscr{R}=\Pi\Theta'=\Pi\mathscr{R}(\mathcal{I},:)$ naturally forms a simplex structure. Thus applying the SP algorithm to $\mathscr{R}$'s rows with $K$ vertices gets $\mathscr{R}(\mathcal{I},:)$. Note that since $\mathscr{R}(\mathcal{I},:)$ is a $K\times J$ matrix with rank $K$ (recall that $K\ll J$ in this paper), its inverse does not exist but the inverse of $\mathscr{R}(\mathcal{I},:)\mathscr{R}(\mathcal{I},:)'$ exists. Set $Z=\mathscr{R}\mathscr{R}(\mathcal{I},:)'(\mathscr{R}(\mathcal{I},:)\mathscr{R}'(\mathcal{I},:))^{-1}$, we have $Z\equiv \Pi$ since $\mathscr{R}=\Pi\mathscr{R}(\mathcal{I},:)$. The following algorithm named Ideal GoM-SRM summarizes the above analysis. Input: $\mathscr{R}$ and $K$. Output: $\Pi$ and $\Theta$.
\begin{itemize}
  \item Apply SP to $\mathscr{R}$'s rows with $K$ vertices to get $\mathscr{R}(\mathcal{I},:)$.
  \item Set $Z=\mathscr{R}\mathscr{R}(\mathcal{I},:)'(\mathscr{R}(\mathcal{I},:)\mathscr{R}'(\mathcal{I},:))^{-1}$.
  \item Recover $\Pi$ and $\Theta$ by the last two steps of Ideal GoM-SRSC.
\end{itemize}
\begin{algorithm}
\caption{\textbf{GoM-SRM}}
\label{alg:SRM}
\begin{algorithmic}[1]
\Require $R$ and $K$.
\Ensure $\hat{\Pi}$ and $\hat{\Theta}$.
\State Apply SP algorithm to $R$'s rows with $K$ vertices to get $R(\hat{\mathcal{I}},:)$.
\State $\hat{Z}_{}=\mathrm{max}(0,RR(\hat{\mathcal{I}},:)'(R(\hat{\mathcal{I}},:)R'(\hat{\mathcal{I}},:))^{-1})$.
\State Obtain $\hat{\Pi}$ and $\hat{\Theta}$ by steps 5 and 6 of Algorithm \ref{alg:SRSC}.
\end{algorithmic}
\end{algorithm}

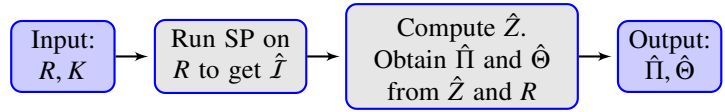
\begin{figure}[H]
\centering
\begin{tikzpicture}
[auto,
decision/.style={diamond, draw=blue, thick, fill=blue!20,
    text width=4.5em,align=flush center,
    inner sep=1pt},
block/.style ={rectangle, draw=blue, thick, fill=gray!20,
    text width=8em,align=center, rounded corners,
    minimum height=2em},
block1/.style ={diamond, draw=blue, thick, fill=orange!20,
    text width=1.4em,align=center, rounded corners,
    minimum height=2em},
block2/.style ={rectangle, draw=blue, thick, fill=blue!20,
    text width=3.2em,align=center, rounded corners,
    minimum height=2em},
block33/.style ={rectangle, draw=blue, thick, fill=gray!20,
    text width=6em,align=center, rounded corners,
    minimum height=2em},
block3/.style ={rectangle, draw=blue, thick, fill=gray!20,
    text width=8em,align=center, rounded corners,
    minimum height=2em},
block4/.style ={rectangle, draw=blue, thick, fill=gray!20,
    text width=5em,align=center, rounded corners,
    minimum height=2em},
line/.style ={draw, thick, -latex',shorten >=2pt},
cloud/.style ={draw=red, thick, ellipse,fill=red!20,
    minimum height=2em}]
\matrix [column sep=5mm,row sep=7mm]
{
&\node [block2] (input) {Input: $R, K$};
&\node [block4] (vh) {Run SP on $R$ to get $\hat{\mathcal{I}}$};
&\node [block] (mr) {Compute $\hat{Z}$. Obtain $\hat{\Pi}$ and $\hat{\Theta}$ from $\hat{Z}$ and $R$};
&\node [block2] (output) {Output: $\hat{\Pi}, \hat{\Theta}$};\\
};
\begin{scope}[every path/.style=line]
\path (input) -- (vh);
\path (vh) -- (mr);
\path (mr) -- (output);
\end{scope}
\end{tikzpicture}
\caption{Flowchart of Algorithm \ref{alg:SRM}.}\label{IllusAlgorithm4}
\end{figure}

The GoM-SRM algorithm (i.e., Algorithm \ref{alg:SRM}) is the real case of the Ideal GoM-SRM algorithm and its flowchart is shown in Fig.~\ref{IllusAlgorithm4}. Compared with GoM-SRSC, GoM-CRSC, and GoM-SSC, there is no SVD step in GoM-SRM. GoM-SRM's complexity is $O(NJK)$. Since the complexities of GoM-SRSC, GoM-CRSC, and GoM-SSC are $O(\mathrm{max}(N^{2}, J^{2})K)$, GoM-SRM runs faster than the other three methods when $N\neq J$.
\begin{rem}
Recall that GoM-SSC is a refinement of Algorithm 2 in \citep{chen2024spectral}, originally designed for categorical data with binary responses, to now handle categorical data with polytomous responses. This makes GoM-SSC a predecessor algorithm for comparison. We have excluded horizontal comparisons with the JML algorithm \citep{sirt_3.13-194} \deleted{due to the fact}\added{because} that Algorithm 2 \citep{chen2024spectral}, the simplified version of GoM-SSC, outperforms JML in both accuracy and efficiency \citep{chen2024spectral}. Furthermore, the JML algorithm is limited to categorical data with binary responses, and extending it to handle polytomous responses is a fascinating and promising future direction.
\end{rem}
\subsection{Evaluation metrics}
In our simulation studies, we assess the performance of our algorithms in accurately estimating the membership matrix $\Pi$, the item parameter matrix $\Theta$, and the number of latent classes $K$ using \added{the} three metrics below. For the estimation of $\Pi$ in simulated categorical data with known $\Pi$, we calculate the Hamming error, defined as the minimum over all $K\times K$ permutation matrices $\mathcal{P}$ of the relative $\ell_1$ norm difference between $\hat{\Pi}$ and $\Pi\mathcal{P}$, i.e., $\mathrm{min}_{\mathcal{P}\in \mathcal{S}}\frac{1}{N}\|\hat{\Pi}-\Pi\mathcal{P}\|_1$, where $\mathcal{S}$ denotes the set collecting all $K\times K$ permutation matrices. This metric, ranging between 0 and 1, quantifies the discrepancy between $\Pi$ and its estimation $\hat{\Pi}$ with smaller values indicating better estimation accuracy.

For the estimation of $\Theta$, we utilize the Relative error, which is defined as the minimum over all $K\times K$ permutation matrices $\mathcal{P}$ of the relative Frobenius norm difference between $\hat{\Theta}$ and $\Theta\mathcal{P}$, i.e., $\mathrm{min}_{\mathcal{P}\in \mathcal{S}}\frac{\|\hat{\Theta}-\Theta\mathcal{P}\|_F}{\|\Theta\|_F}$. This metric measures the discrepancy between $\Theta$ and its estimation $\hat{\Theta}$ with smaller values indicating better estimation accuracy.

Finally, to measure the performance of algorithms in inferring $K$, we calculate the accuracy rate, defined as the fraction of times an algorithm correctly identifies $K$ out of a total number of independent trials. This metric, ranging between 0 and 1, quantifies the success rate of algorithms in estimating $K$ with larger values indicating higher accuracy. By using these three metrics, we can objectively evaluate the performance of our proposed algorithms in estimating $\Pi$, $\Theta$, and $K$ in a controlled simulation environment.
\subsection{Simulations}\label{SecSim}
We test the effectiveness and accuracies of our methods by investigating their sensitivities to the regularization parameter $\tau$, the sparsity parameter $\rho$, and the number of subjects $N$ in this part. Unless specified, for all simulations, we set $J=\frac{N}{4}, K=3$, and $M=4$ (i.e., $R(i,j)\in\{0,1,2,3,4\}$ for $i\in[N], j\in[J]$). Let each latent class have $N_{0}$ pure subjects. Let the top $3N_{0}$ subjects $\{1,2,\ldots, N_{0}\}$ be pure and the rest $(N-3N_{0})$ subjects be mixed. For mixed subject $i$, we set its membership score $\Pi(i,:)$ as $\Pi(i,:)=(r_{1}, r_{2}, 1-r_{1}-r_{2})$, where $r_{1}=\frac{\mathrm{rand}(1)}{2}$, $r_{1}=\frac{\mathrm{rand}(1)}{2}$, and $\mathrm{rand}(1)$ is a random number generated from a Uniform distribution on $[0,1]$. Let $N_{0}=\frac{N}{4}$ in all simulations. Set $\bar{B}$ as an $J\times K$ matrix such that its $(j,k)$-th entry is $\mathrm{rand}(1)$. Set $B=\frac{\bar{B}}{\mathrm{max}_{j\in[J], k\in[K]}\bar{B}(j,k)}$ (in this way, $B$'s maximum entry is 1). Except for the case where we study the choice of $\tau$, we set $\tau$ as its default value $M\mathrm{max}(N,J)$. $\rho$ and $N$ are set independently for each experiment. For each parameter setting, we report the averaged metric over 100 repetitions.
\begin{figure}
\centering
\resizebox{\columnwidth}{!}{
\subfigure[Hamming error against $\tau$]{\includegraphics[width=0.2\textwidth]{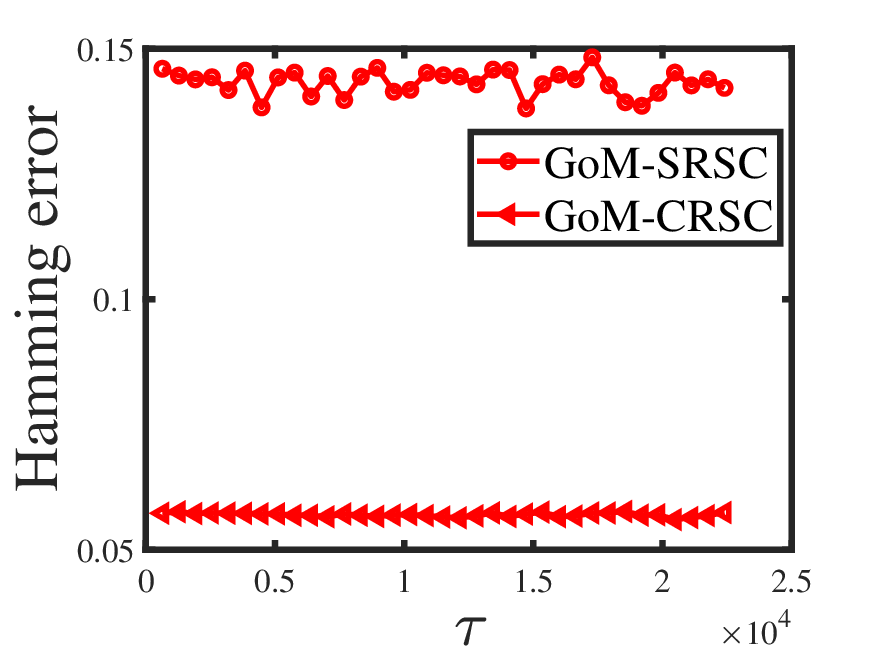}}
\subfigure[Relative error against $\tau$]{\includegraphics[width=0.2\textwidth]{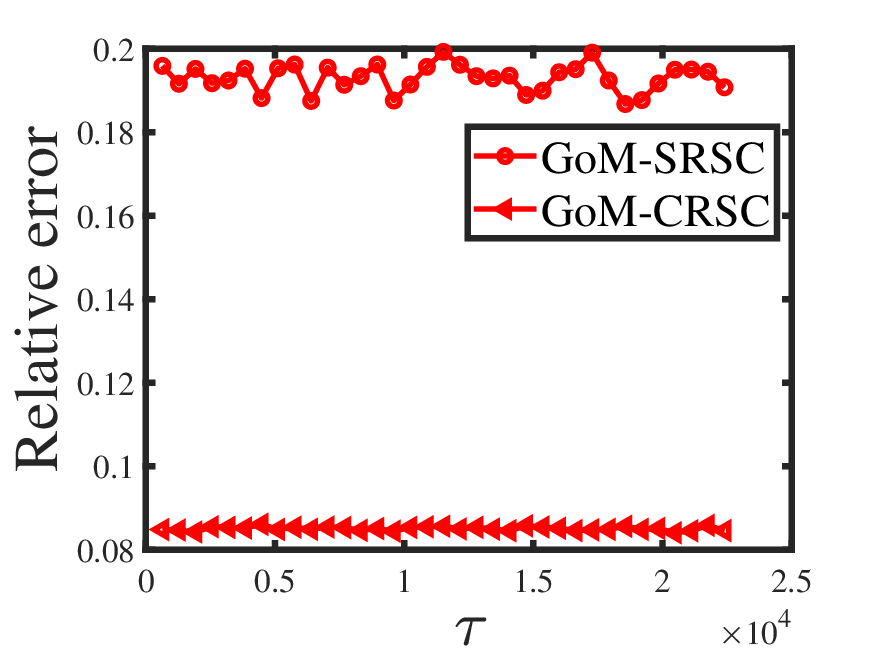}}
\subfigure[Running time against $\tau$]{\includegraphics[width=0.2\textwidth]{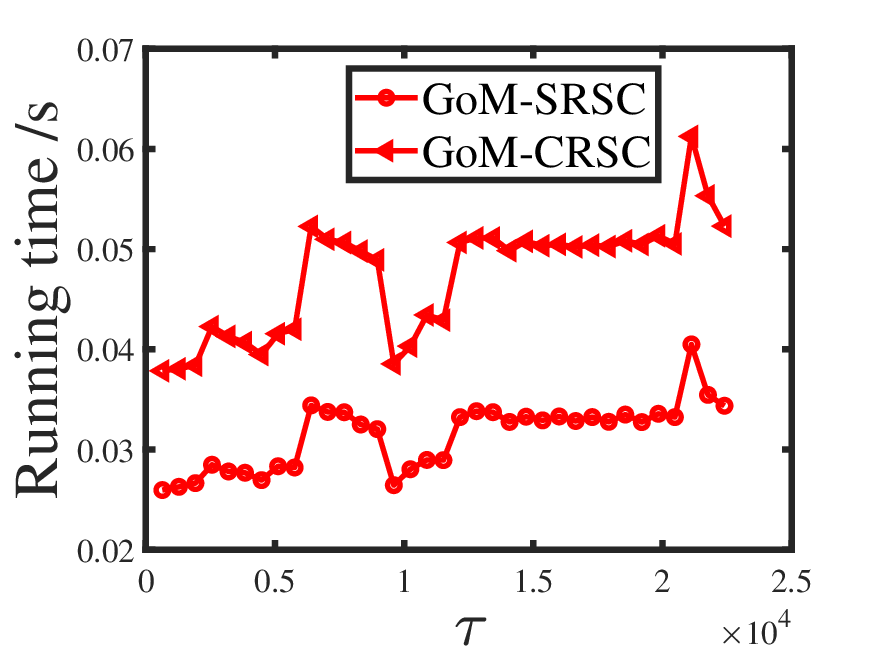}}
\subfigure[Accuracy against $\tau$]{\includegraphics[width=0.2\textwidth]{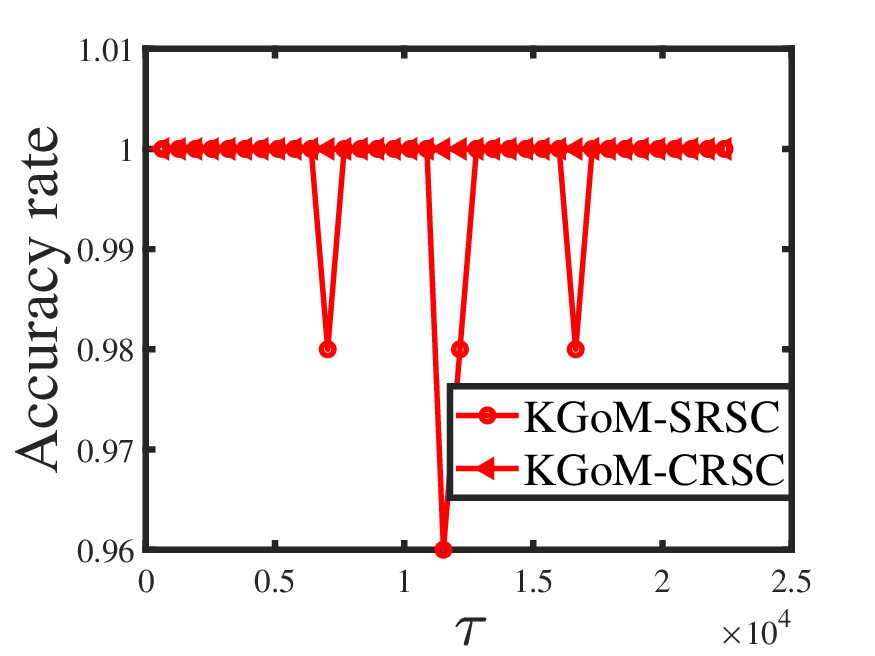}}
}
\caption{Results of Experiment 1.}
\label{Ex1} 
\end{figure}

\texttt{Experiment 1: Changing $\tau$.} Recall that there is a regularization parameter $\tau$ in our algorithms GoM-SRSC and GoM-CRSC, in this experiment, we investigate the effect of $\tau$ empirically. We set $(\rho, N)=(1,800)$ and $\tau=\alpha M\mathrm{max}(N, J)$. We vary $\alpha$ in the range $\{0.2, 0.4, \ldots, 7\}$. The results, displayed in Fig.~\ref{Ex1}, indicate that (1) our GoM-SRSC and GoM-CRSC are insensitive to the choice of the regularizer $\tau$. In this paper we choose the default value of $\tau$ as $M\mathrm{max}(N,J)$ because it matches our theoretical analysis and our methods are insensitive to $\tau$; (2) GoM-CRSC returns more accurate estimations of $\Pi$ and $\Theta$ than GoM-SRSC; (3) GoM-SRSC runs slightly faster than GoM-CRSC; (4) Our KGoM-SRSC and KGoM-CRSC have high accuracies in estimating $K$, indicating the effectiveness of our metric computed by Equation (\ref{fuzzyModularity}) in measuring the quality of latent mixed membership analysis.

\begin{figure}
\centering
\resizebox{\columnwidth}{!}{
\subfigure[Hamming error against $\rho$]{\includegraphics[width=0.2\textwidth]{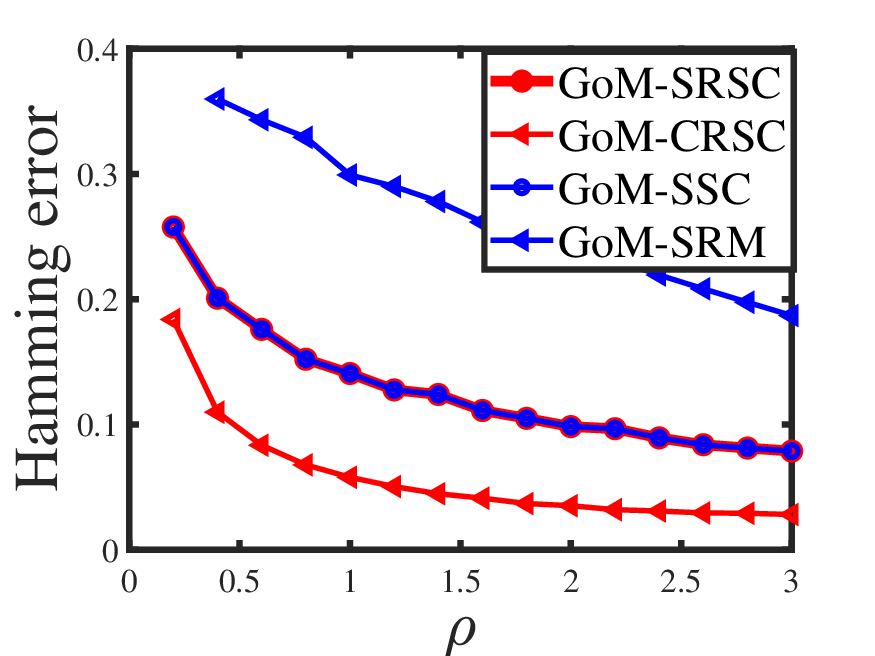}}
\subfigure[Relative error against $\rho$]{\includegraphics[width=0.2\textwidth]{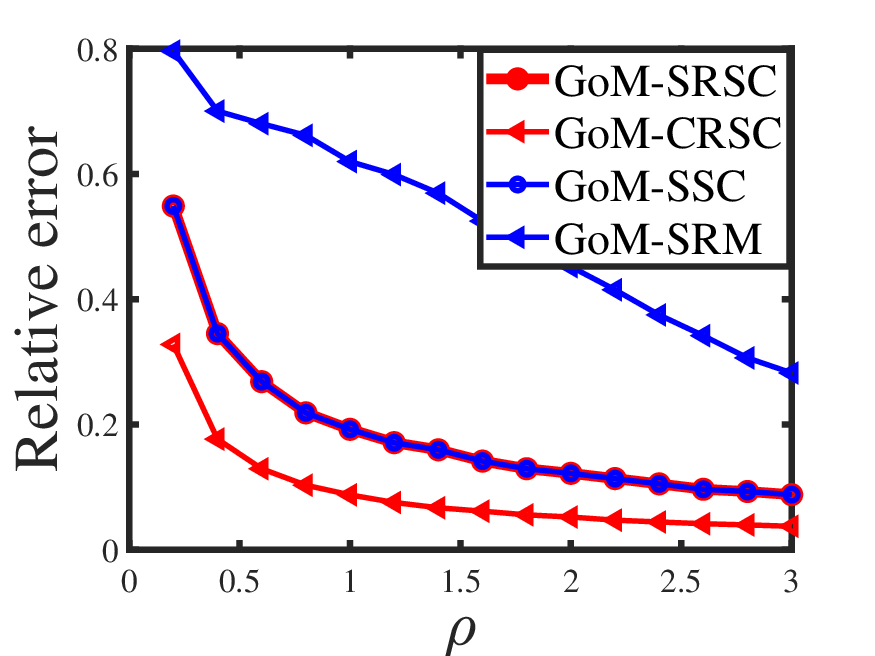}}
\subfigure[Running time against $\rho$]{\includegraphics[width=0.2\textwidth]{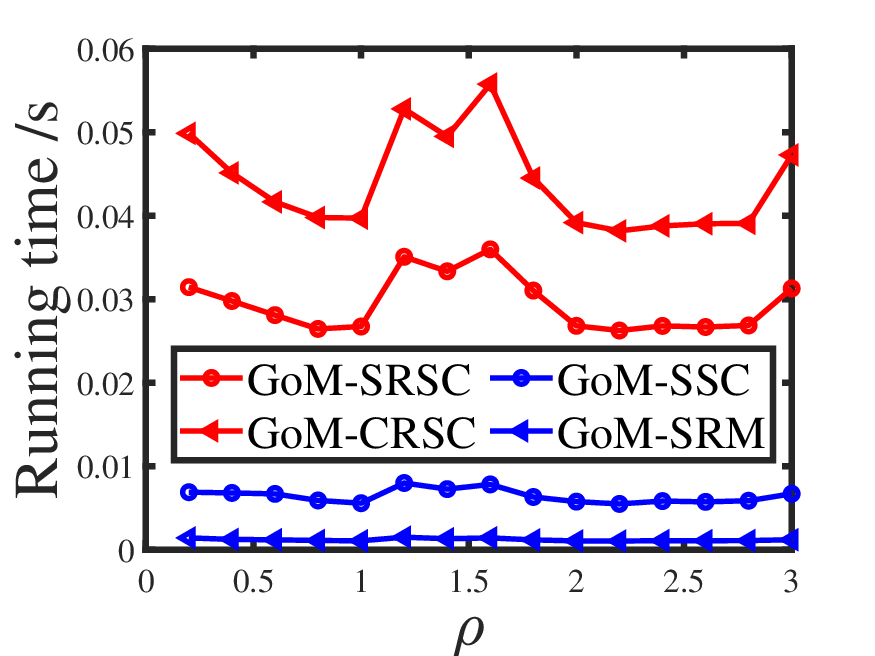}}
\subfigure[Accuracy against $\rho$]{\includegraphics[width=0.2\textwidth]{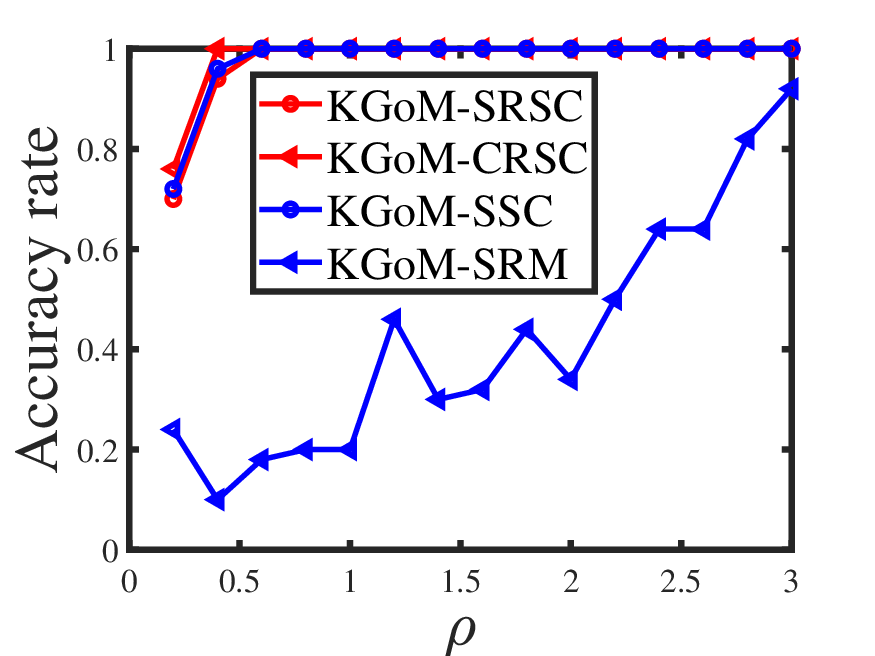}}
}
\caption{Results of Experiment 2.}
\label{Ex2} 
\end{figure}
\texttt{Experiment 2: Changing $\rho$.} Here, we investigate the sensitivity of our methods to the sparsity parameter $\rho$. Let $N=800$ and $\rho$ range in $\{0.2, 0.4, \ldots, 3\}$. Results are shown in Fig.~\ref{Ex2}. We have the following conclusions: (1) as we expect, our methods have better performances in estimating $\Pi$ and $\Theta$ as $\rho$ increases, and this verifies our analysis after Theorem \ref{Main}; (2) for the task of estimating $\Pi$ and $\Theta$,  GoM-CRSC outperforms the other three methods, GoM-SRSC enjoys similar performances as GoM-SSC, and GoM-SRM performs poorest. For computing time, GoM-SRM runs fastest and this supports our computational cost analysis for GoM-SRM; (3) for the task of inferring $K$, our KGoM-SRSC, KGoM-CRSC, and KGoM-SSC enjoy high accuracies while KGoM-SRM performs poorest. Again, the high accuracy of KGoM-SRSC, KGoM-CRSC, and KGoM-SSC guarantees the high effectiveness of our metric in quantitating the strength of latent mixed membership analysis.

\begin{figure}
\centering
\resizebox{\columnwidth}{!}{
\subfigure[Hamming error against $N$]{\includegraphics[width=0.2\textwidth]{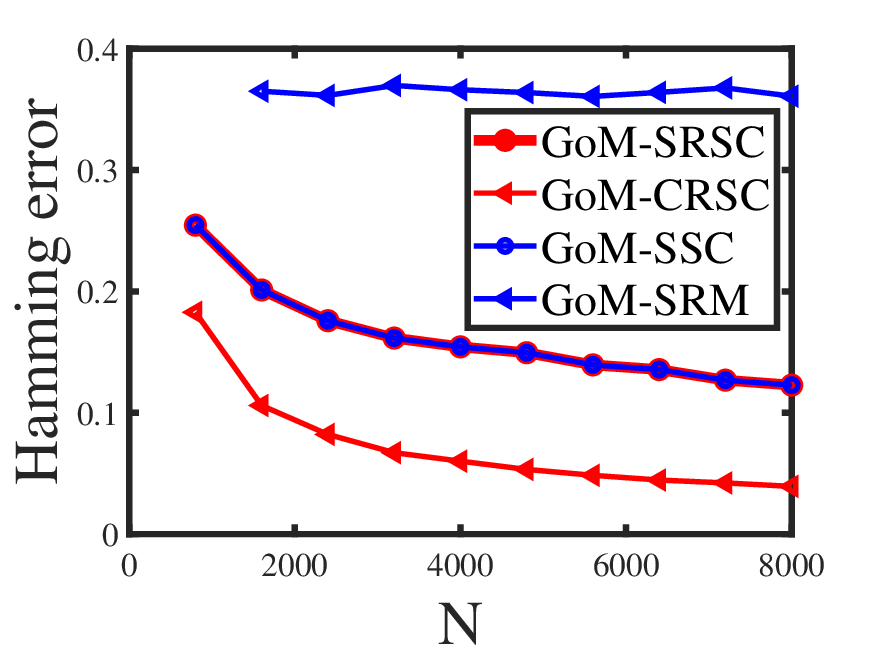}}
\subfigure[Relative error against $N$]{\includegraphics[width=0.2\textwidth]{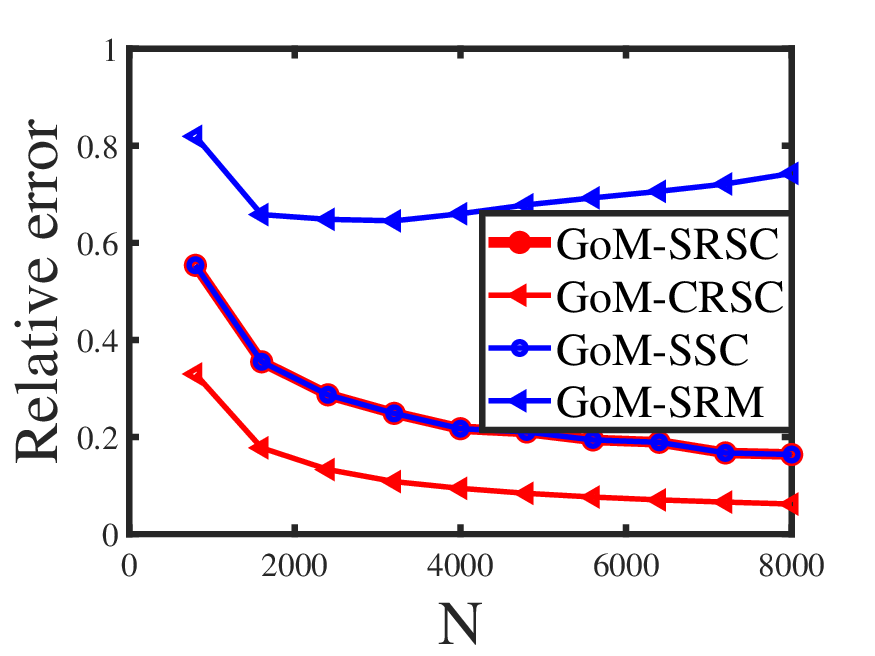}}
\subfigure[Running time against $N$]{\includegraphics[width=0.2\textwidth]{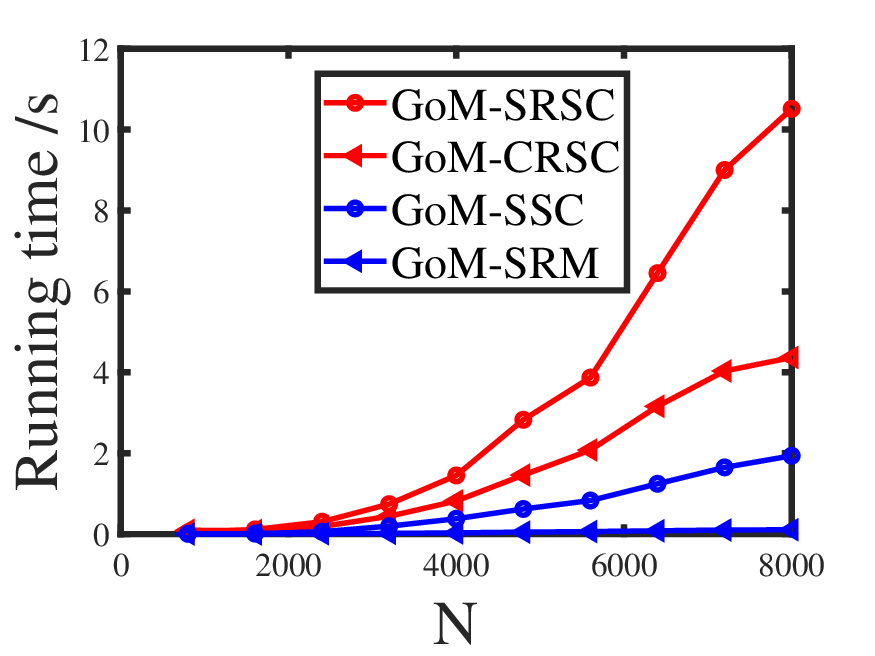}}
\subfigure[Accuracy against $N$]{\includegraphics[width=0.2\textwidth]{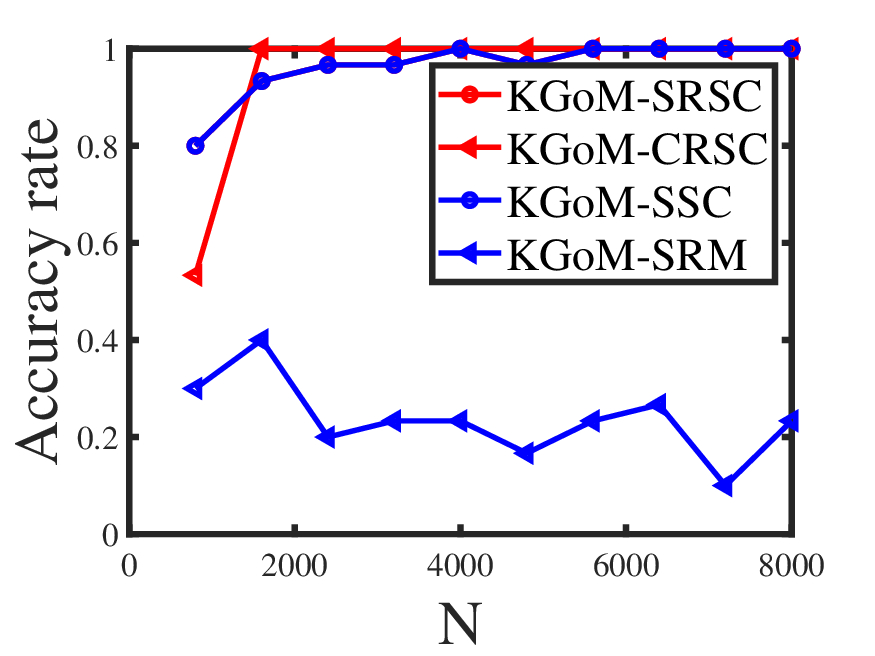}}
}
\caption{Results of Experiment 3.}
\label{Ex3} 
\end{figure}

\texttt{Experiment 3: Changing $N$.} Let $\rho=0.2$ and $N$ take values from the set $\{800, 1600, \ldots, 8000\}$. Fig.~\ref{Ex3} shows the results. We see that (1) GoM-SRSC, GoM-CRSC, and GoM-SCC behave better as $N$ and $J$ increase, which verifies our analysis after Theorem \ref{Main}; (2) GoM-CRSC has the best performances in estimating $\Pi$ and $\Theta$ though GoM-CRSC runs slowest; (3) GoM-SRSC performs similarly to GoM-SCC and GoM-SRM performs poorest; (4) Our methods process categorical data with 8000 subjects and 2000 items within 12 seconds; (5) Our KGoM-SRSC, KGoM-CRSC, and KGoM-SSC enjoy satisfactory performances in estimate $K$, indicating the powerfulness of our metric.

\texttt{Experiment 4: A toy example.} Set $K=2, N=20, J=10$, and $M=5$. Set $\Pi$ and $\Theta$ as the 1st and 2nd matrices in Fig.~\ref{Ex4simR}, respectively. The last matrix of Fig.~\ref{Ex4simR} displays a response matrix $R$ generated from the  GoM model under the above setting. We apply our algorithms to this $R$ and Table \ref{ErrorRatesSimulatedR} records the numerical results. We see that our GoM-CRSC performs best in estimating $\Pi$ and $\Theta$, GoM-SRSC outperforms GoM-SCC slightly, and all methods estimate $K$ correctly for this toy example. Meanwhile, applying our methods to the $R$ in the last matrix of Fig.~\ref{Ex4simR} gets the estimations of $\Pi$ and $\Theta$, and we show these estimations in Fig.~\ref{Ex4simePieTheta}. As we expect, since the Hamming error and Relative error are non-zero for all methods by Table \ref{ErrorRatesSimulatedR}, $\hat{\Pi}$ (and $\hat{\Theta}$) differs slightly from $\Pi$ (and $\Theta$) for each method. We also find that estimating the memberships for mixed subjects is more challenging than that of pure subjects for all methods. Finally, as suggested by our theoretical analysis after Theorem \ref{Main} and numerical analysis for Experiment 3, to make the estimations more accurate, we should increase $N$ and $J$. However, we set $N=20$ and $J=10$ for this experiment mainly for the visualizations of $\Pi, \Theta, R, \hat{\Pi},$ and $\hat{\Theta}$.
\begin{figure}
\centering
\includegraphics[width=1\textwidth]{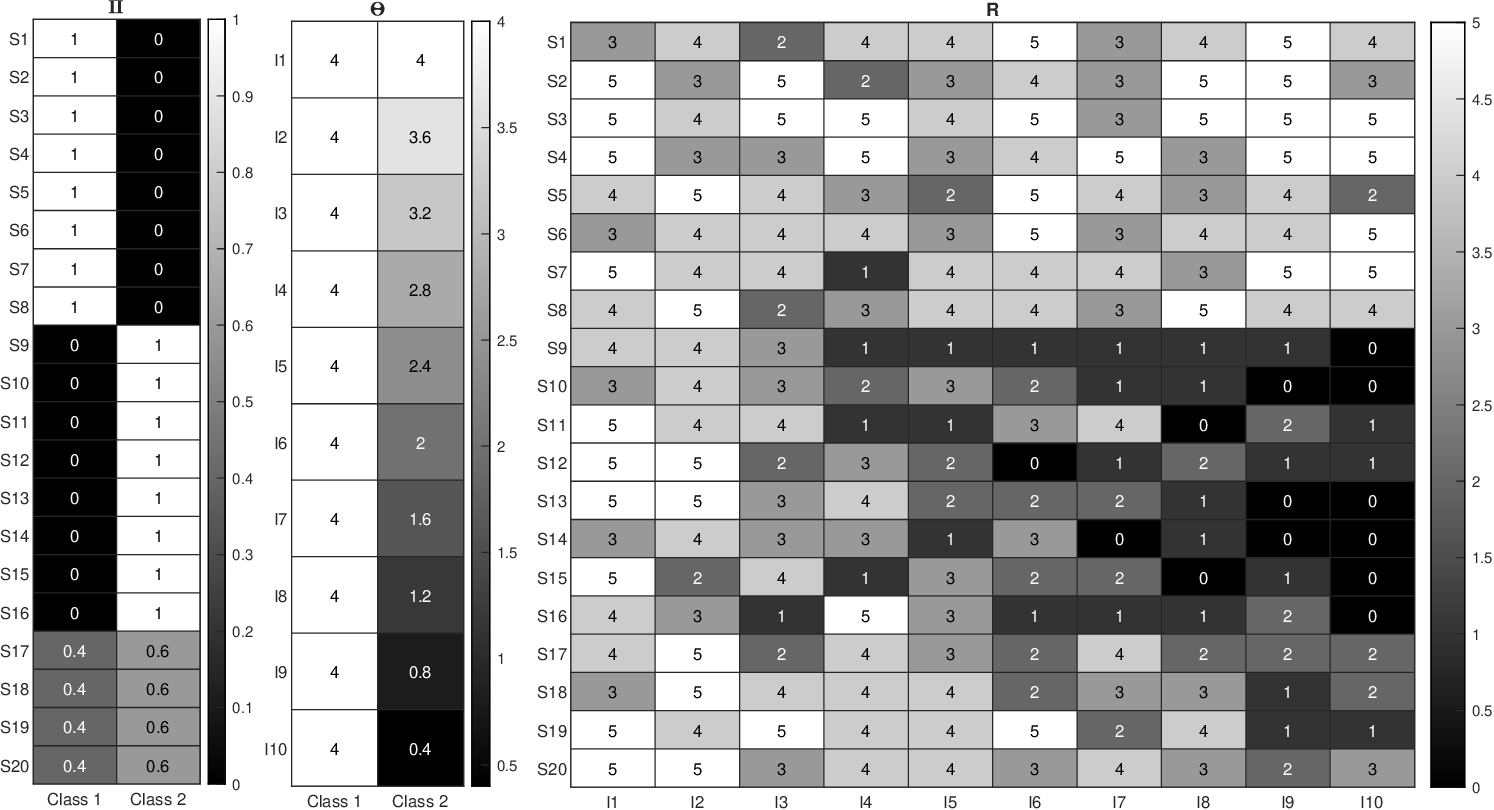}
\caption{We set $\Pi$ and $\Theta$ for Experiment 4 as the 1st and 2nd matrices, respectively. The 3rd matrix shows a $R$ generated from the GoM model in Experiment 4. Here, S$i$ denotes subject $i$ and I$j$ means item $j$.}
\label{Ex4simR} 
\end{figure}

\begin{figure}
\centering
\includegraphics[width=1\textwidth]{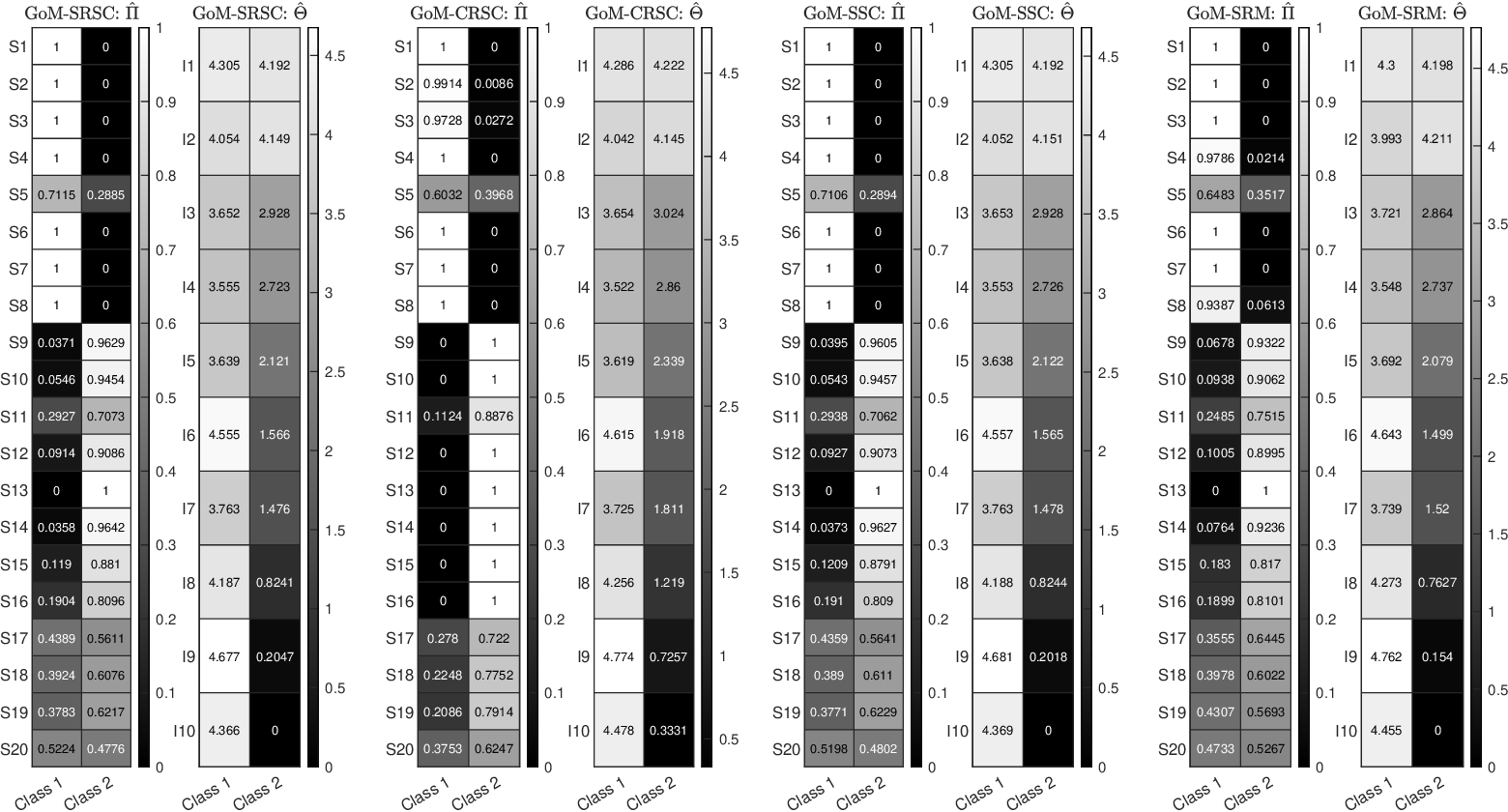}
\caption{Estimated $\Pi$ and $\Theta$ of each method for Experiment 4. Here, S$i$ denotes subject $i$ and I$j$ means item $j$.}
\label{Ex4simePieTheta} 
\end{figure}

\begin{table}[h!]
\footnotesize
	\centering
	\caption{Results of Experiment 4.}
	\label{ErrorRatesSimulatedR}
	\begin{tabular}{cccccccccccc}
\hline\hline&Hamming error&Relative error&Estimated $K$\\
\hline
GoM-SRSC&0.0650&0.1142&2\\
GoM-CRSC&0.0529&0.1034&2\\
GoM-SCC&0.0654&0.1146&2\\
GoM-SRM&0.0772&0.1251&2\\
\hline\hline
\end{tabular}
\end{table}
\section{Application to real-world categorical data}\label{sec7}
We apply our methods to real categorical data in this part. Table \ref{realdata} summarizes basic information for real categorical data considered in this paper, where  $\varsigma=\frac{\mathrm{Number~of~zero~elements~of~}R}{NJ}$ denotes the proportion of no-responses in $R$ and we have removed those subjects that have no response to all items for all data.  The true mixed memberships and number of latent classes are unknown for all real data used here. These datasets can be downloaded from \url{http://konect.cc/networks/movielens-100k_rating/} (accessed
on January 1, 2024) and \url{https://openpsychometrics.org/_rawdata/} (accessed
on January 1, 2024).

\added{For MovieLens 100k, users rate movies, with a higher score indicating greater preference. For IPIP, responses range from 1 (Strongly disagree) to 5 (Strongly agree), with 3 representing neutrality. Similarly, KIMS uses a 1-5 scale, where 1 signifies 'Never or very rarely true' and 5 'Very often or always true'. For NPI, 1 and 2 represent different choices for each question. Across all datasets, 0 signifies no response. Comprehensive details on the statements for IPIP, KIMS, and NPI are presented in Fig.~\ref{IPIPPiTheta}, Fig.~\ref{KIMSTheta}, and Fig.~\ref{NPITheta}, respectively.} \deleted{For MovieLens 100k, users rate movies and a higher rating score means more likes. For IPIP, 1=Strongly disagree, 2=Disagree, 3=Neither agree nor disagree, 4=Agree, and 5=Strongly agree. For KIMS, 1=Never or very rarely true, 2=Rarely true, 3=Sometimes true, 4=Often true, and 5=Very often or always true. For NPI, 1 means one choice of a question and 2 means another choice. For all data, 0 means no response. Details of all statements for IPIP, KIMS, and NPI can be found in Fig.~\ref{IPIPPiTheta}, Fig.~\ref{KIMSTheta}, and Fig.~\ref{NPITheta}, respectively.}

By Table \ref{realdata}, we see that the MovieLens 100k data is the sparsest while there are only a few no-responses for the other three datasets. This can be understood in the following way. Because MovieLens 100k is a user-rating-movie network and the number of users is smaller than that of movies, a user usually rates just a few movies. For comparison, because the other three datasets are test data and the number of individuals is much larger than that of questions, an individual usually answers all questions and only a few individuals have no responses to some questions.

To determine how many latent classes one should use for each data, we run our algorithms to $R$ to compute the modularity by Equation (\ref{fuzzyModularity}). Based on the results reported in Table \ref{RealDataModularityK}, we observe that (1) our GoM-CRSC always returns the largest modularity for all data, which suggests that the estimated number of latent classes for MovieLens 100k, IPIP, KIMS, and NPI should be 3, 2, 2, and 2, respectively; (2) GoM-SRSC performs similarly to GoM-SSC; (3) GoM-SRM returns poorest partitions of latent mixed memberships since its modularity is the smallest for each data. These observations are consistent with our numerical findings in Experiments 2-4. From now on, we only consider our GoM-CRSC for real data analysis since its modularity is the largest for each data.
\begin{table}[h!]
	\centering
	\caption{Basic information of real-world categorical data used in this paper.}
	\label{realdata}
	\resizebox{\columnwidth}{!}{
	\begin{tabular}{cccccccccccc}
\hline\hline&Subject meaning&Item meaning&$N$&$J$&$M$&$\varsigma$\\
\hline
MovieLens 100k&User&Movie&943&1682&5&93.7\%\\
International Personality Item Pool (IPIP) personality test&Individual&Statement&1004&40&5&0.6\%\\
Kentucky Inventory of Mindfulness Skills (KIMS) test&Individual&Statement&601&39&5&0.59\%\\
Narcissistic Personality Inventory (NPI) test&Individual&Question&11241&40&2&0.3\%\\
\hline\hline
\end{tabular}}
\end{table}

\begin{table}[h!]
\footnotesize
	\centering
	\caption{The $K$ maximizing the modularity and the respective modularity of each method for data in Table \ref{realdata}.}
	\label{RealDataModularityK}
	\begin{tabular}{cccccccccccc}
\hline\hline
Dataset&GoM-SRSC&GoM-CRSC&GoM-SSC&GoM-SRM\\
\hline
MovieLens 100k&(2, 0.0461)&(3, 0.0730)&(2, 0.0501)&(2,0176)\\
IPIP&(4, 0.0021)&(2, 0.0067)&(4, 0.0021)&(4,0.0019)\\
KIMS&(3, 0.0027)&(2, 0.0081)&(3, 0.0027)&(3, 0.0024)\\
NPI&(4, 0.0017)&(2, 0.0054)&(4, 0.0017)&(14, 0.00028)\\
\hline\hline
\end{tabular}
\end{table}

Let $\hat{\Pi}$ and $\hat{\Theta}$ be the estimated membership matrix and the estimated item parameter matrix returned by applying GoM-CRSC to real data with $K$ latent classes, where $K$ is the one returned by KGoM-CRSC. To simplify our analysis, we name the $K$ latent classes as Class 1, Class 2, $\ldots$, and Class $K$ for each data.
\begin{defin}\label{munuPhi}
Call subject $i$ a highly pure subject if $\mathrm{max}_{k\in[K]}\hat{\Pi}(i,k)\geq0.9$, a highly mixed subject if $\mathrm{max}_{k\in[K]}\hat{\Pi}(i,k)\leq0.7$ for $i\in[N]$, and a moderate pure subject otherwise. Define $\mu=\frac{|\{i\in[N]: \mathrm{max}_{k\in[K]}\hat{\Pi}(i,k)\geq0.9\}|}{N}$ and $\nu=\frac{|\{i\in[N]: \mathrm{max}_{k\in[K]}\hat{\Pi}(i,k)\leq0.7\}|}{N}$.
\end{defin}
Definition \ref{munuPhi} says that $\mu$ ($\nu$) is the proportion of highly pure (mixed) subjects. Results reported in Table \ref{RealDataResults} show that though most subjects are highly pure, there still exist many highly mixed subjects for each data. By Table \ref{RealDataResults}, we see that GoM-CRSC processes data with 11241 subjects and 40 items within 2 seconds, which suggests its efficiency.

Recall that $\mathscr{R}=\Pi\Theta'$ under GoM, we have  $\mathscr{R}(i,j)=\sum_{k\in[K]}\Pi(i,k)\Theta(j,k)$, which suggests that a larger $\Theta(j,k)$ gives a larger $\mathscr{R}(i,j)$ given the membership score $\Pi(i,:)$ of subject $i$, i.e., a larger $\Theta(j,k)$ tends to make $R(i,j)$ larger since $\mathscr{R}(i,j)$ is the expectation of $R(i,j)$ under GoM. Based on this observation, next, we interpret the estimated latent classes for each data briefly.
\begin{table}[h!]
\footnotesize
	\centering
	\caption{$\mu, \nu$, and runtime when applying GoM-CRSC to real  data considered in this paper.}
	\label{RealDataResults}
\begin{tabular}{cccccccccc}
\hline\hline
dataset&$\mu$&$\nu$&Runtime\\
\hline
MovieLens 100k&0.4602&0.2333&0.1077s\\
IPIP&0.6036&0.2112&0.0467s\\
KIMS&0.6040&0.1947&0.0345s\\
NPI&0.6249&0.1882&1.5903s\\
\hline\hline
\end{tabular}
\end{table}

For MovieLens 100k, Fig.~\ref{100KPiTheta} displays a subset of $\hat{\Pi}$ and a subset of $\hat{\Theta}$ for better visualization. From the estimated mixed memberships of subjects $2, 3, 4, 15, 17, 21, 26, 27$, and 30, we see that these subjects are highly pure and they are more likely to belong to Class 3. Meanwhile, subjects 12, 20, and 28 are more likely to belong to Class 1. From the right matrix of Fig.~\ref{100KPiTheta}, we find that $\hat{\Theta}(j,1)$ is usually larger than $\hat{\Theta}(j,2)$ while $\hat{\Theta}(j,2)$ is usually larger than $\hat{\Theta}(j,3)$ for $j\in[30]$. In fact, we have $\sum_{j\in[1682]}\hat{\Theta}(j,1)=704.6280, \sum_{j\in[1682]}\hat{\Theta}(j,2)=552.2857$, and  $\sum_{j\in[1682]}\hat{\Theta}(j,3)=178.6997$. Based on the above observation, Class 1 can be interpreted as users with an optimistic view of movies, Class 2 can be interpreted as users with a neutral view of movies, and Class 3 can be interpreted as users with a passive view of movies. Users in Class 1 are more likely to give higher ratings to movies than users in Class 2 and Class 3.

Fig.~\ref{RealPi100k} depicts the ternary diagram of $\hat{\Pi}$ for the MovieLens 100k dataset. This diagram provides a visual representation of the purity of each user within the dataset. Users exhibiting varying degrees of purity are distinctively represented by different colors and shapes, allowing for a direct assessment of their purity levels. It is noteworthy that the positions of the shapes within the ternary diagram carry significant information. Specifically, the location of the shape corresponding to user $i$, who exhibits a higher value of $\mathrm{max}_{k\in[K]}\hat{\Pi}(i,k)$, tends to be closer to one of the vertexes of the triangle for all $i\in[N]$. Conversely, the shapes representing users who exhibit a high degree of mixing are positioned towards the interior of the triangle. This spatial arrangement within the ternary diagram effectively illustrates the varying purity levels among users in the MovieLens 100k dataset.

\begin{figure}
\centering
\includegraphics[width=0.7\textwidth]{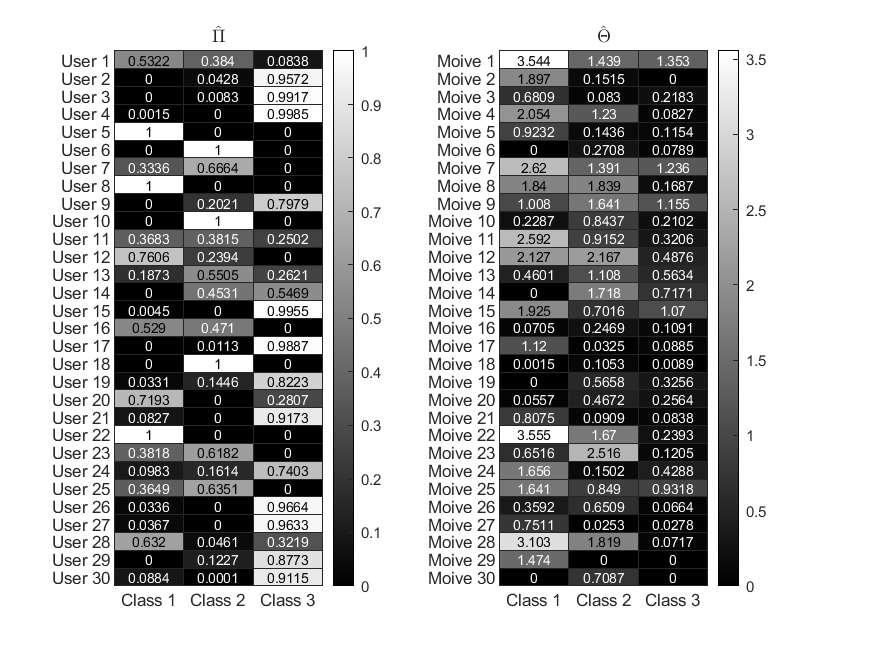}
\caption{Left matrix: heatmap of the estimated mixed memberships of the first 30 users for the MovieLens 100k data. Right matrix: heatmap of the estimated item parameters of the first 30 items for the MovieLens 100k data.}
\label{100KPiTheta} 
\end{figure}

\begin{figure}
\centering
\includegraphics[width=0.8\textwidth]{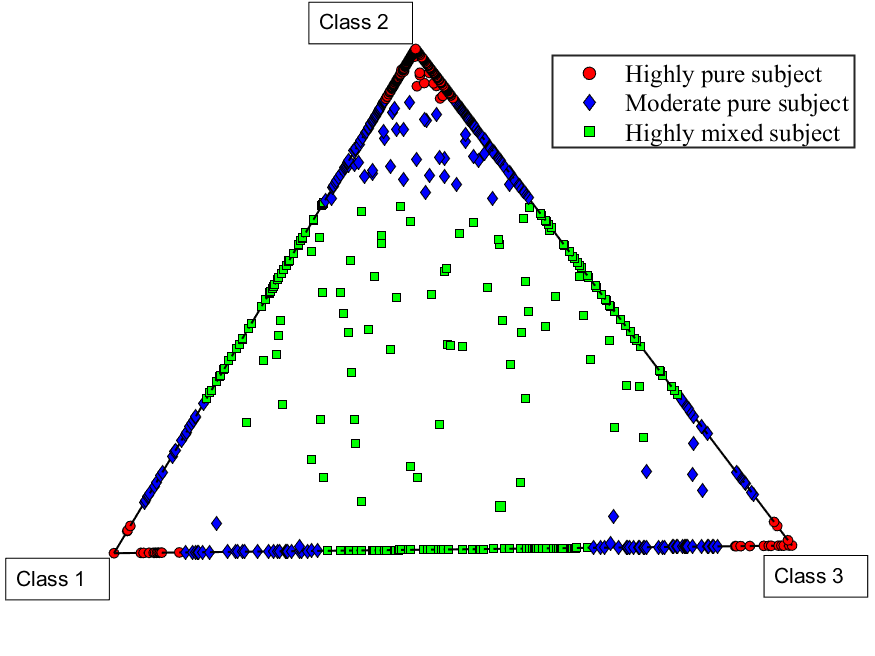}
\caption{Ternary diagram of the $943\times3$ estimated membership matrix $\hat{\Pi}$ for the MovieLens 100k data. Within this diagram, each distinct shape (dot, square, and diamond) represents an individual user, and its specific location within the triangle corresponds to the user's three membership scores.}
\label{RealPi100k} 
\end{figure}

\begin{figure}
\centering
\includegraphics[width=0.5\textwidth]{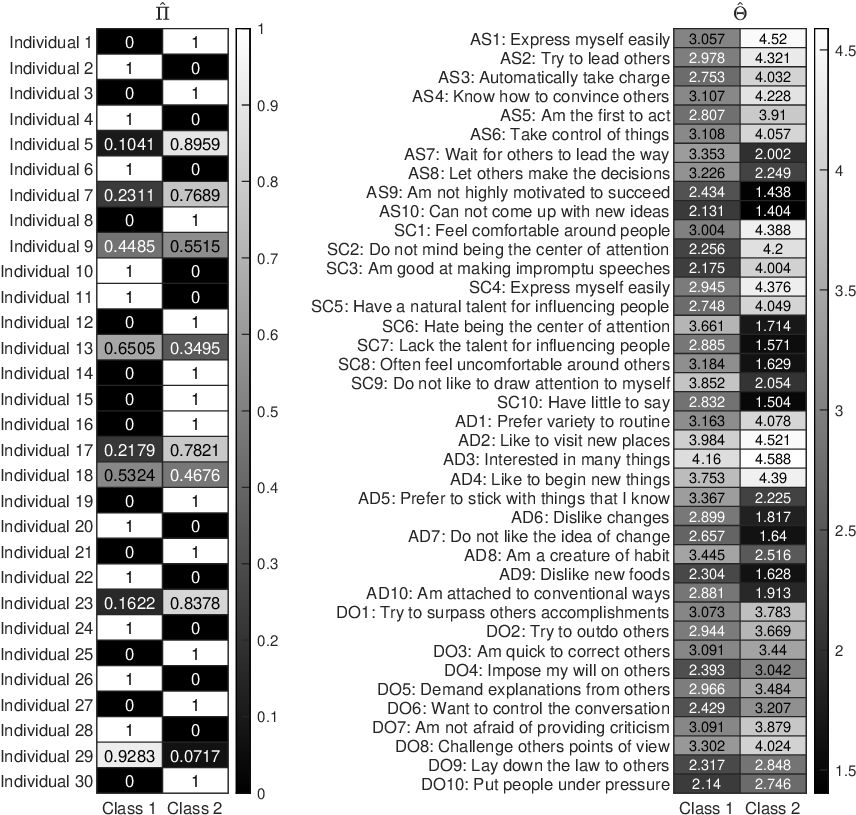}
\caption{Left matrix: the estimated mixed memberships of the first 30 individuals for the IPIP data. Right matrix: heatmap of $\hat{\Theta}$ for the IPIP data, where AS, SC, AD, and DO mean four personality factors assertiveness, social confidence, adventurousness, and dominance, respectively.}
\label{IPIPPiTheta} 
\end{figure}

\begin{figure}
\centering
\resizebox{\columnwidth}{!}{
\subfigure[$\hat{\Pi}$: IPIP]{\includegraphics[width=0.4\textwidth]{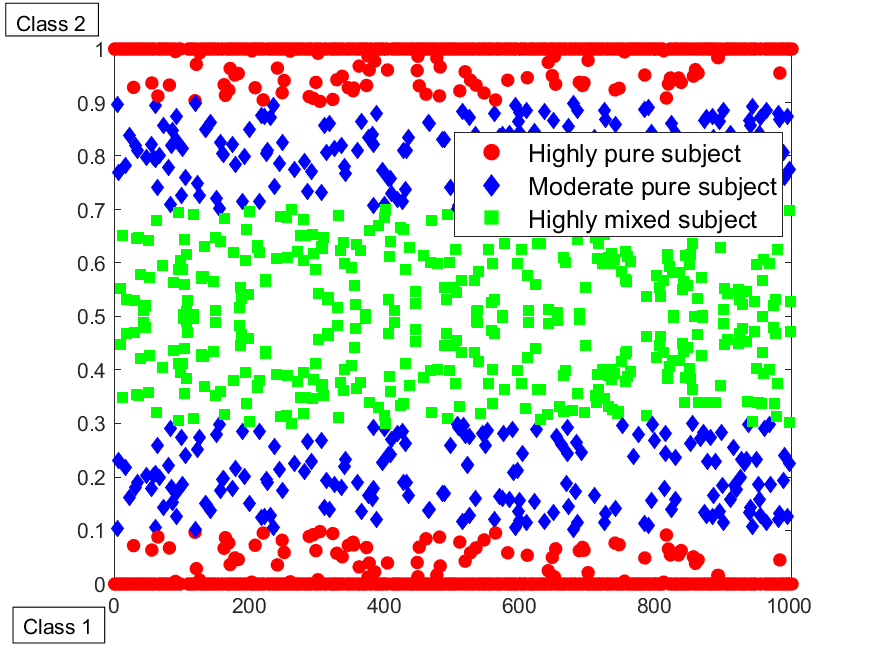}}
\subfigure[$\hat{\Pi}$: KIMS]{\includegraphics[width=0.4\textwidth]{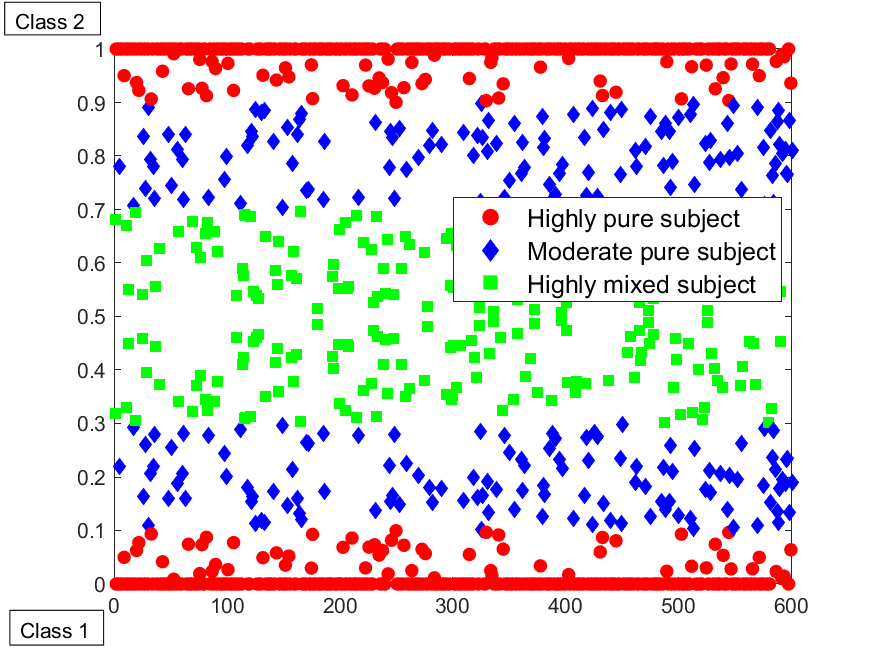}}
}
\caption{Panel (a): binary diagram of the $1004\times2$ estimated membership matrix $\hat{\Pi}$ for the IPIP data. Panel (b): binary diagram of the $601\times2$ estimated membership matrix $\hat{\Pi}$ for the KIMS data. Within both diagrams, each distinct shape (dot, square, and diamond) represents an individual, and its specific location within the interval $[0,1]$ corresponds to the individual's two membership scores.}
\label{RealPiIPIPKIMS} 
\end{figure}

For IPIP, the estimated mixed memberships of the first 30 individuals and the estimated item parameter matrix $\hat{\Theta}$ are displayed in Fig.~\ref{IPIPPiTheta}. From $\hat{\Pi}$, we can find subjects' memberships. Recall that for the IPIP data, a higher response value to a statement means a stronger agreement, by analyzing $\hat{\Theta}$, we interpret Class 1 as people who are socially passive and Class 2 as socially optimistic. People in Class 2 are usually more assertive, adventurous, and dominant than those in Class 1.
\begin{figure}
\centering
\includegraphics[width=0.7\textwidth]{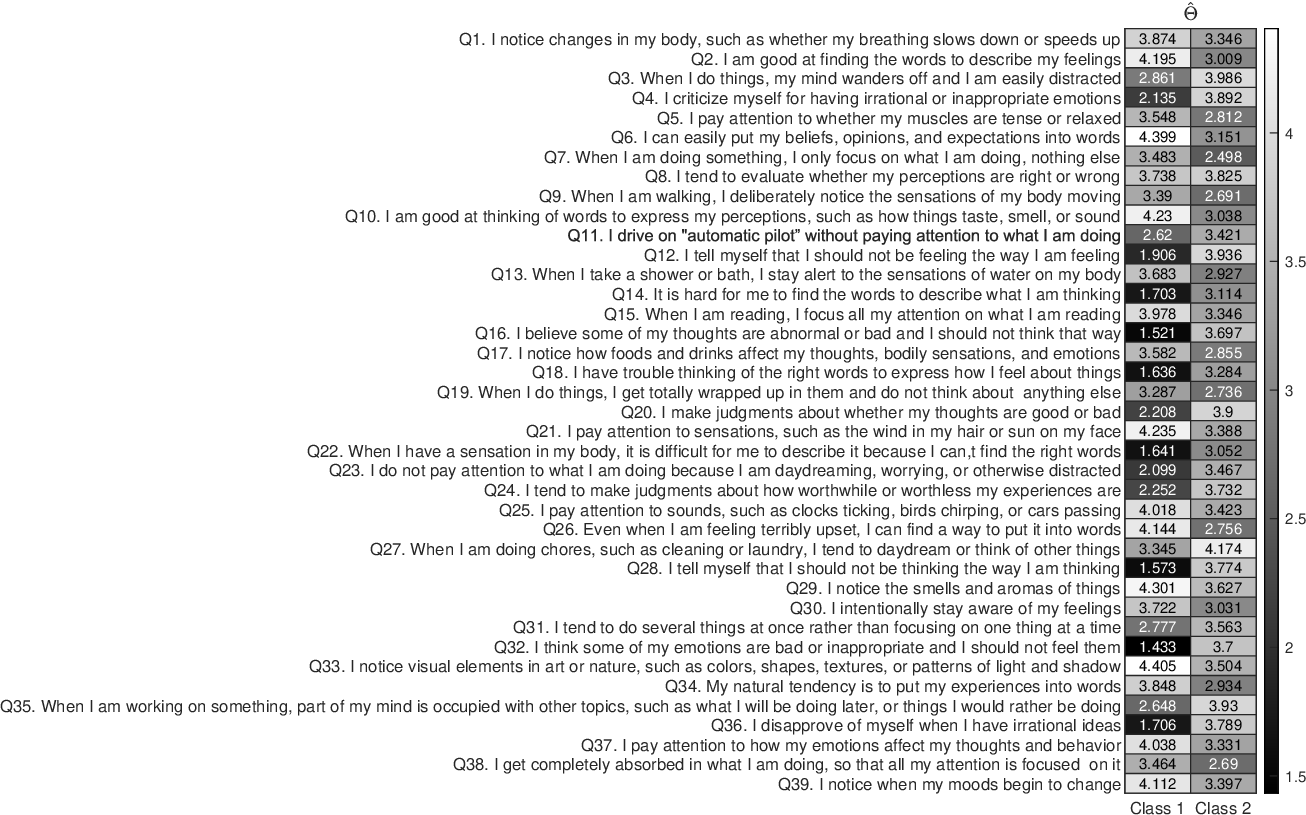}
\caption{Heatmap of $\hat{\Theta}$ for the KIMS data.}
\label{KIMSTheta} 
\end{figure}
\begin{figure}
\centering
\includegraphics[width=0.7\textwidth]{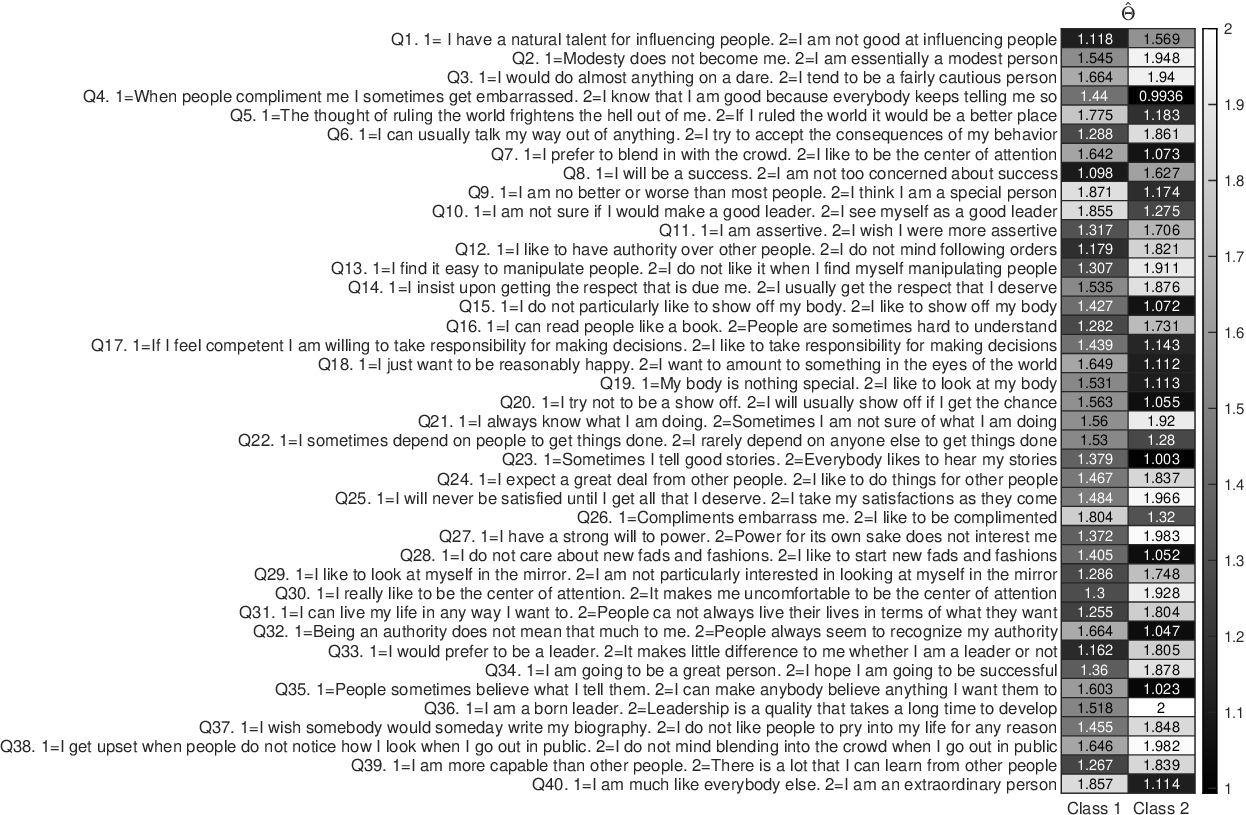}
\caption{Heatmap of $\hat{\Theta}$ for the NPI data.}
\label{NPITheta} 
\end{figure}

\added{For KIMS, Fig.~\ref{KIMSTheta} offers a heatmap visualization of $\hat{\Theta}$. Based on this data, we classify individuals into two distinct categories: Class 1, representing mindful individuals, and Class 2, comprising those who are mindless. Notably, individuals belonging to Class 1 tend to exhibit a higher degree of self-focus compared to those in Class 2. Additionally, Fig.~\ref{RealPiIPIPKIMS} depicts the binary diagrams of $\hat{\Pi}$ for both the IPIP and KIMS datasets. The interpretation of this figure closely resembles the explanation provided for Fig.~\ref{RealPi100k}, therefore, a detailed discussion is omitted here for the sake of brevity.} \deleted{For KIMS, Fig.~\ref{KIMSTheta} shows $\hat{\Theta}$'s heatmap. For this data, we interpret Class 1 as people who are mindful and Class 2 as people who are mindless. People in Class 1 are more focused on himself/herself than those in Class 2. Fig.~\ref{RealPiIPIPKIMS} presents the binary diagrams of $\hat{\Pi}$ for the IPIP dataset and the KIMS dataset. The interpretation of this figure aligns closely with the explanation provided for Fig.~\ref{RealPi100k}, hence a detailed discussion is omitted here for brevity.}

\added{Fig.~\ref{NPITheta} presents a heatmap depiction of $\hat{\Theta}$ for the NPI dataset. Within this context, Class 1 represents individuals who exhibit narcissistic traits, while Class 2 signifies those who are more modest. Notably, individuals categorized under Class 1 exhibit a higher degree of narcissism compared to those belonging to Class 2.} \deleted{Fig.~\ref{NPITheta} displays $\hat{\Theta}$'s heatmap for the NPI data. For this data, Class 1 can be interpreted as people who are narcissistic and Class 2 as people who are modest. People in Class 1 are more narcissistic than those in Class 2.}
\section{Conclusion}\label{sec8}
In this paper, we present two novel algorithms for mixed membership estimation and item parameter estimation under the Grade of Membership (GoM) model for categorical data with polytomous responses. These algorithms not only offer remarkable efficiency but also consistently yield accurate parameter estimations under mild conditions on data sparsity. To assess the quality of the latent mixed membership analysis in categorical data, we have developed a metric that computes fuzzy modularity based on the product of the response matrix and its transpose. By combining this metric with our algorithms, we can determine the number of latent classes $K$ for categorical data generated from the GoM model. Our experimental findings highlight the superior performance of our algorithms in estimating the mixed membership matrix, the item parameter matrix, and the number of latent classes under the GoM model. Notably, our methods for estimating $K$ exhibit remarkable accuracy, further validating the utility of our metric in evaluating the quality of the latent mixed membership analysis. This research represents a leap forward in comprehending mixed membership estimation and item parameter estimation under the GoM model for categorical data. The proposed algorithms and metric offer practical tools for analyzing categorical data and evaluating the quality of latent mixed membership analysis, with the potential to revolutionize latent class analysis and enable researchers to gain deeper insights into the structures of categorical data.
\section*{CRediT authorship contribution statement}
\textbf{Huan Qing} is the sole author of this article.
\section*{Declaration of competing interest}
The author declares no competing interests.
\section*{Data availability}
Data and code will be made available on request.
\section*{Acknowledgements}
H.Q. was sponsored by \added{the Scientific Research Foundation of Chongqing University of Technology (Grant No: 0102240003) and the} Natural Science Foundation of Chongqing, China (Grant No: CSTB2023NSCQ-LZX0048).
\appendix
\section{Proofs under GoM}\label{SecProofs}
\subsection{Proof of Lemma \ref{SVDPopulationLtau}}
\begin{proof}
For the 1st statement, $\mathscr{R}=\Pi\Theta'\rightarrow \mathscr{R}'=\Theta\Pi'\rightarrow\mathscr{R}'\Pi=\Theta\Pi'\Pi\rightarrow\Theta=\mathscr{R}'\Pi(\Pi'\Pi)^{-1}$, where the $K\times K$ matrix $\Pi'\Pi$ is nonsingular since $\Pi$'s rank is $K$ when Condition (C1) hold.

For the 2nd statement, $\mathscr{L}_{\tau}=\mathscr{D}^{-1/2}_{\tau}\mathscr{R}=\mathscr{D}^{-1/2}_{\tau}\Pi\Theta'=U\Sigma V'$ under GoM gives $U_{\tau}=\mathscr{D}^{1/2}_{\tau}U=\Pi\Theta'V\Sigma^{-1}$. By $\Pi(\mathcal{I},:)=I_{K\times K}$, we have $U_{\tau}(\mathcal{I},:)=\Pi(\mathcal{I},:)\Theta'V\Sigma^{-1}=\Theta'V\Sigma^{-1}$, i.e., $U_{\tau}=\Pi U_{\tau}(\mathcal{I},:)$. Sure, when $\Pi(i,:)=\Pi(\bar{i},:)$ for two distinct subjects $i$ and $\bar{i}$, we get $U_{\tau}(i,:)=U_{\tau}(\bar{i},:)$.

For the 3rd statement, since $U=\mathscr{D}^{-1/2}_{\tau}\Pi\Theta'V\Sigma^{-1}$, we have $U(i,:)=\mathscr{D}^{-1/2}_{\tau}(i,i)\Pi(i,:)\Theta'V\Sigma^{-1}$, which gives that $U_{*}(i,:)=\frac{\Pi(i,:)\Theta'V\Sigma^{-1}}{\|\Pi(i,:)\Theta'V\Sigma^{-1}\|_{F}}=\frac{\Pi(i,:)U_{\tau}(\mathcal{I},:)}{\|\Pi(i,:)U_{\tau}(\mathcal{I},:)\|_{F}}=\frac{\Pi(i,:)\mathscr{D}^{1/2}_{\tau}(\mathcal{I},\mathcal{I})U(\mathcal{I},:)}{\|\Pi(i,:)U_{\tau}(\mathcal{I},:)\|_{F}}=\frac{\Pi(i,:)\mathscr{D}^{1/2}_{\tau}(\mathcal{I},\mathcal{I})D^{-1}_{U}(\mathcal{I},\mathcal{I})D_{U}(\mathcal{I},\mathcal{I})U(\mathcal{I},:)}{\|\Pi(i,:)U_{\tau}(\mathcal{I},:)\|_{F}}=\frac{\Pi(i,:)\mathscr{D}^{1/2}_{\tau}(\mathcal{I},\mathcal{I})D^{-1}_{U}(\mathcal{I},\mathcal{I})U_{*}(\mathcal{I},:)}{\|\Pi(i,:)U_{\tau}(\mathcal{I},:)\|_{F}}$ for $i\in[N]$. Sure, when $\Pi(i,:)=\Pi(\bar{i},:)$, we have $U_{*}(i,:)=U_{*}(\bar{i},:)$.
\end{proof}
\subsection{Proof of Lemma \ref{ConeCon}}
\begin{proof}
By the proof of Lemma \ref{SVDPopulationLtau}, we know that $U=\mathscr{D}^{-1/2}_{\tau}\Pi U_{\tau}(\mathcal{I},:)=\mathscr{D}^{-1/2}_{\tau}\Pi\mathscr{D}^{1/2}_{\tau}(\mathcal{I},\mathcal{I})U(\mathcal{I},:)$. By $I_{K\times K}=U'U$, we have $(U(\mathcal{I},:)U'(\mathcal{I},:))^{-1}=\mathscr{D}^{1/2}_{\tau}(\mathcal{I},\mathcal{I})\Pi'\mathscr{D}^{-1}_{\tau}\Pi\mathscr{D}^{1/2}_{\tau}(\mathcal{I},\mathcal{I})$. Then we have $(U_{*}(\mathcal{I},:)U'_{*}(\mathcal{I},:))^{-1}=(N_{U}(\mathcal{I},\mathcal{I})U(\mathcal{I},:)U'(\mathcal{I},:)N_{U}(\mathcal{I},\mathcal{I}))^{-1}=N^{-1}_{U}(\mathcal{I},\mathcal{I})\mathscr{D}^{1/2}_{\tau}(\mathcal{I},\mathcal{I})\Pi'\mathscr{D}^{-1}_{\tau}\Pi\mathscr{D}^{1/2}_{\tau}(\mathcal{I},\mathcal{I})N^{-1}_{U}(\mathcal{I},\mathcal{I})$. Since all entries of $N_{U}, \mathscr{D}_{\tau}$, and $\Pi$ are nonnegative, we have $(U_{*}(\mathcal{I},:)U'_{*}(\mathcal{I},:))^{-1}\mathbf{1}>0$.
\end{proof}
\subsection{Proof of Lemma \ref{ZstarPi}}
\begin{proof}
Since $Y=D_{o}\Pi\mathscr{D}^{1/2}_{\tau}(\mathcal{I},\mathcal{I})D^{-1}_{U}(\mathcal{I},\mathcal{I})$ by Lemma \ref{SVDPopulationLtau}, we have $D_{o}\Pi=YD_{U}(\mathcal{I},\mathcal{I})\mathscr{D}^{-1/2}_{\tau}(\mathcal{I},\mathcal{I})=U_{*}U^{-1}_{*}(\mathcal{I},:)D_{U}(\mathcal{I},\mathcal{I})\mathscr{D}^{-1/2}_{\tau}(\mathcal{I},\mathcal{I})=D_{U}UU^{-1}_{*}(\mathcal{I},:)D_{U}(\mathcal{I},\mathcal{I})\mathscr{D}^{-1/2}_{\tau}(\mathcal{I},\mathcal{I})$, which gives that $D^{-1}_{U}D_{0}\Pi=UU^{-1}_{*}(\mathcal{I},:)D_{U}(\mathcal{I},\mathcal{I})\mathscr{D}^{-1/2}_{\tau}(\mathcal{I},\mathcal{I})\equiv Z_{*}$. Because $D^{-1}_{U}D_{0}$ is a diagonal matrix, we have $\Pi(i,:)=\frac{Z_{*}(i,:)}{\|Z_{*}(i,:)\|_{1}}$ for $i\in[N]$.
\end{proof}
\subsection{Proof of Theorem \ref{Main}}
\begin{proof}
Let $\lambda_{k}(X)$ denote the $k$-th largest eigenvalue in magnitude , $\kappa(X)$ denote the conditional number,  $\|X\|$ denote the spectral norm, and $\|X\|_{2\rightarrow\infty}$ denote the maximum $l_{2}$-norm among all rows for any matrix $X$ in this paper. Set $\delta_{\mathrm{min}}=\mathrm{min}_{i\in[N]}\mathscr{D}(i,i)$ and $\delta_{\mathrm{max}}=\mathrm{max}_{i\in[N]}\mathscr{D}(i,i)$. To prove this theorem, we need Lemma \ref{BoundUiF} and Lemma \ref{rowwiseerror} below.
\begin{lem}\label{BoundUiF}
Under $\mathrm{GoM}(\Pi,\Theta)$, we have below results:
\begin{itemize}
  \item $\sqrt{\frac{\tau+\delta_{\mathrm{min}}}{\tau+\delta_{\mathrm{max}}}}\frac{1}{\sqrt{K\lambda_{1}(\Pi'\Pi)}}\leq \|U(i,:)\|_{F}\leq\sqrt{\frac{\tau+\delta_{\mathrm{max}}}{\tau+\delta_{\mathrm{min}}}}\frac{1}{\sqrt{\lambda_{K}(\Pi'\Pi)}}\mathrm{~for~}i\in[N]$.
  \item $\frac{\tau+\delta_{\mathrm{min}}}{\lambda_{1}(\Pi'\Pi)}\leq\lambda_{K}(U_{\tau}(\mathcal{I},:)U'_{\tau}(\mathcal{I},:))\leq\lambda_{1}(U_{\tau}(\mathcal{I},:)U'_{\tau}(\mathcal{I},:))\leq\frac{\tau+\delta_{\mathrm{max}}}{\lambda_{K}(\Pi'\Pi)}$.
  \item $\frac{\tau+\delta_{\mathrm{min}}}{(\tau+\delta_{\mathrm{max}})\kappa(\Pi'\Pi)}\leq\lambda_{K}(U_{*}(\mathcal{I},:)U'_{*}(\mathcal{I},:))\leq\lambda_{1}(U_{*}(\mathcal{I},:)U'_{*}(\mathcal{I},:))\leq K\frac{\tau+\delta_{\mathrm{max}}}{\tau+\delta_{\mathrm{min}}}\kappa(\Pi'\Pi)$.
  \item $\frac{\rho^{2}\lambda_{K}(\Pi'\Pi)\lambda_{K}(B'B)}{\tau+\delta_{\mathrm{max}}}\leq \lambda_{K}(\mathscr{L}_{\tau}\mathscr{L}'_{\tau})\leq \lambda_{1}(\mathscr{L}_{\tau}\mathscr{L}'_{\tau})\leq\frac{\rho^{2}\lambda_{1}(\Pi'\Pi)\lambda_{1}(B'B)}{\tau+\delta_{\mathrm{min}}}$.
  \item $\|V(j,:)\|_{F}\leq\frac{K^{1.5}\kappa(\Pi)}{\sigma_{K}(B)}\sqrt{\frac{\tau+\delta_{\mathrm{max}}}{\tau+\delta_{\mathrm{min}}}}$ for $j\in[J]$.
\end{itemize}
\end{lem}
\begin{proof}
The first three bullets are guaranteed by Lemmas C.2 and C.3 of \citep{qing2023regularized}. For the fourth statement, since $\mathscr{L}_{\tau}=\mathscr{D}^{-1/2}_{\tau}\mathscr{R}=\mathscr{D}^{-1/2}_{\tau}\Pi\Theta'=\rho \mathscr{D}^{-1/2}_{\tau}\Pi B'$, we have
\begin{align*}
\lambda_{K}(\mathscr{L}_{\tau}\mathscr{L}'_{\tau})=\rho^{2}\lambda_{K}(\mathscr{D}^{-1/2}_{\tau}\Pi B'B\Pi'\mathscr{D}^{-1/2}_{\tau})=\rho^{2}\lambda_{K}(\mathscr{D}^{-1}_{\tau}\Pi B'B\Pi')\geq\rho^{2}\lambda_{K}(\mathscr{D}^{-1}_{\tau})\lambda_{K}(\Pi B'B\Pi')\geq\frac{\rho^{2}\lambda_{K}(\Pi'\Pi)\lambda_{K}(B'B)}{\tau+\delta_{\mathrm{max}}},
\end{align*}
and
\begin{align*}
\lambda_{1}(\mathscr{L}_{\tau}\mathscr{L}'_{\tau})=\rho^{2}\lambda_{1}(\mathscr{D}^{-1/2}_{\tau}\Pi B'B\Pi'\mathscr{D}^{-1/2}_{\tau})=\rho^{2}\lambda_{1}(\mathscr{D}^{-1}_{\tau}\Pi B'B\Pi')\leq\rho^{2}\lambda_{1}(\mathscr{D}^{-1}_{\tau})\lambda_{1}(\Pi B'B\Pi')\leq\frac{\rho^{2}\lambda_{1}(\Pi'\Pi)\lambda_{1}(B'B)}{\tau+\delta_{\mathrm{min}}}.
\end{align*}

For the last statement, since $\mathscr{L}_{\tau}=\mathscr{D}^{-1/2}_{\tau}\mathscr{R}=\mathscr{D}^{-1/2}_{\tau}\Pi\Theta'=\rho\mathscr{D}^{-1/2}_{\tau}\Pi B'=U\Sigma V'$, we have $V=\rho B\Pi'\mathscr{D}^{-1/2}_{\tau}U\Sigma^{-1}$, which gives that
\begin{align*}
\|V(j,:)\|_{F}&=\rho \|B(j,:)\Pi'\mathscr{D}^{-1/2}_{\tau}U\Sigma^{-1}\|_{F}\leq\rho \|B(j,:)\|_{F}\|\Pi'\mathscr{D}^{-1/2}_{\tau}U\|_{F}\|\Sigma^{-1}\|_{F}\leq \rho\sqrt{K}\|\Pi'\mathscr{D}^{-1/2}_{\tau}\|_{F}\|\Sigma^{-1}\|_{F}\\
&\leq\frac{\rho K\|\mathscr{D}^{-1/2}_{\tau}\Pi\|_{F}}{\sigma_{K}(\mathscr{L}_{\tau})}\leq\frac{\rho K\|\Pi\|_{F}}{\sigma_{K}(\mathscr{L}_{\tau})\sqrt{\tau+\delta_{\mathrm{min}}}}\leq\frac{\rho K^{1.5}\|\Pi\|}{\sigma_{K}(\mathscr{L}_{\tau})\sqrt{\tau+\delta_{\mathrm{min}}}}\leq\frac{K^{1.5}\kappa(\Pi)}{\sigma_{K}(B)}\sqrt{\frac{\tau+\delta_{\mathrm{max}}}{\tau+\delta_{\mathrm{min}}}},
\end{align*}
where the last inequality holds by the 4th statement of this lemme.
\end{proof}

\begin{lem}\label{rowwiseerror}
Under $\mathrm{GoM}(\Pi,\Theta)$, suppose that Assumption \ref{Assum1} and Condition \ref{Con1} are satisfied, and $\sigma_{K}(\mathscr{L}_{\tau})\gg \sqrt{\frac{\rho N\mathrm{log}(N)}{\tau}}$. With high probability, we have
\begin{align*}	
\varpi:=\|\hat{U}\hat{U}'-UU'\|_{2\rightarrow\infty}=O(\sqrt{\frac{\mathrm{log}(N+J)}{\rho NJ}}).
\end{align*}
\end{lem}
\begin{proof}
Set $H=\hat{U}'U$. Let $H=U_{H}\Sigma_{H}V'_{H}$ be $H$'s top-$K$ SVD. Set $\mathrm{sgn}(H)=U_{H}V'_{H}$. Under $\mathrm{GoM}(\Pi,\Theta)$, for $i\in[N], j\in[J]$, we get $\mathbb{E}(R(i,j)-\mathscr{R}(i,j))=0$, $\mathbb{E}((L_{\tau}(i,j)-\mathscr{L}_{\tau}(i,j))^{2})=\mathbb{E}[(\frac{R(i,j)}{\sqrt{\tau+D(i,i)}}-\frac{\mathscr{R}(i,j)}{\sqrt{\tau+\mathscr{D}(i,i)}})^{2}]\leq \frac{\mathrm{Var}(R(i,j))}{\mathrm{min}(\tau+1,\tau+\delta_{\mathrm{min}})}=\frac{M\frac{\mathscr{R}(i,j)}{M}(1-\frac{\mathscr{R}(i,j)}{M})}{\mathrm{min}(\tau+1,\tau+\delta_{\mathrm{min}})}=\frac{\mathscr{R}(i,j)(1-\frac{\mathscr{R}(i,j)}{M})}{\mathrm{min}(\tau+1,\tau+\delta_{\mathrm{min}})}\leq\frac{\rho}{\mathrm{min}(\tau+1,\tau+\delta_{\mathrm{min}})}=O(\rho/\tau), |L_{\tau}(i,j)-\mathscr{L}_{\tau}(i,j)|\leq \frac{M}{\sqrt{\mathrm{min}(\tau+1,\tau+\delta_{\mathrm{min}})}}=O(M/\sqrt{\tau})$. Meanwhile, since $0<\delta_{\mathrm{min}}\leq\delta_{\mathrm{max}}=\mathrm{max}_{i\in[N]}\sum_{j=1}^{J}\mathscr{R}(i,j)=\rho\mathrm{max}_{i\in[N]}\sum_{j=1}^{J}\Pi(i,:)B'(j,:)\leq \rho J\leq\rho\mathrm{max}(N,J)\leq M\mathrm{max}(N,J)$ and $\frac{\tau+\delta_{\mathrm{max}}}{\tau+\delta_{\mathrm{min}}}$ approximates to 1 as $\tau$ increases, we see that to make $\frac{\tau+\delta_{\mathrm{max}}}{\tau+\delta_{\mathrm{min}}}$ close to 1, a larger $\tau$ is preferred. Suppose that $\tau\geq M\mathrm{max}(N,J)$, we have $\frac{\tau+\delta_{\mathrm{max}}}{\tau+\delta_{\mathrm{min}}}=O(1)$ though $\delta_{\mathrm{max}}\geq\delta_{\mathrm{min}}$. Then, by Lemma \ref{BoundUiF}, Condition \ref{Con1}, and $\tau\geq M\mathrm{max}(N,J)$, we have $\mu=\mathrm{max}(\frac{N\|U\|^{2}_{2\rightarrow\infty}}{K}, \frac{J\|V\|^{2}_{2\rightarrow\infty}}{K})=O(\frac{\tau+\delta_{\mathrm{max}}}{\tau+\delta_{\mathrm{min}}})=O(1)$ and $\kappa(\mathscr{L}_{\tau})=O(1)$. When Assumption \ref{Assum1} is satisfied, we get $c_{b}=\frac{O(M/\sqrt{t})}{\sqrt{\frac{\rho}{\tau}\frac{\mathrm{max}(N,J)}{\mu\mathrm{log}(N+J)}}}\leq O(1)$. Therefore, when $\sigma_{K}(\mathscr{L}_{\tau})\gg \sqrt{\frac{\rho N\mathrm{log}(N)}{\tau}}$ is satisfied, according to Theorem 4 of \cite{chen2021spectral}, with high probability, we have
\begin{align*}
\|\hat{U}\mathrm{sgn}(H)-U\|_{2\rightarrow\infty}\leq O(\frac{\sqrt{K\rho\mathrm{log}(N+J)}}{\sigma_{K}(\mathscr{L}_{\tau})\sqrt{\tau}}).
\end{align*}
By Lemma \ref{BoundUiF}, we have $	\|\hat{U}\mathrm{sgn}(H)-U\|_{2\rightarrow\infty}=O(\frac{1}{\sigma_{K}(\Pi)\sigma_{K}(B)}\sqrt{\frac{K\mathrm{log}(N+J)}{\rho}})$.
$\hat{U}'\hat{U}=I_{K\times K}$ and $U'U=I_{K\times K}$ give $\|\hat{U}\hat{U}'-UU'\|_{2\rightarrow\infty}\leq2\|U-\hat{U}\mathrm{sgn}(H)\|_{2\rightarrow\infty}$. Hence, by Condition \ref{Con1}, we get
\begin{align*}	
\|\hat{U}\hat{U}'-UU'\|_{2\rightarrow\infty}=O(\frac{1}{\sigma_{K}(\Pi)\sigma_{K}(B)}\sqrt{\frac{K\mathrm{log}(N+J)}{\rho}})=O(\sqrt{\frac{\mathrm{log}(N+J)}{\rho NJ}}).
\end{align*}
\end{proof}
For GoM-SRSC, its per-subject error rates of memberships equal that of community detection for mixed networks. Therefore, by the proof of Theorem 1 in \citep{qing2023regularized}, there exist two $K$-by-$K$ permutation matrices $\mathcal{P}$ and $\mathcal{P}_{*}$ such that with high probability,
\begin{align*}
\mathrm{max}_{i\in[N]}\|e'_{i}(\hat{\Pi}-\Pi\mathcal{P})\|_{1}=O(K^{1.5}\varpi\kappa(\Pi)\sqrt{\lambda_{1}(\Pi'\Pi)}) \mathrm{~and~}\mathrm{max}_{i\in[N]}\|e'_{i}(\hat{\Pi}_{*}-\Pi\mathcal{P}_{*})\|_{1}=O(\frac{K^{6.5}\varpi \kappa^{4.5}(\Pi)\lambda_{1}^{1.5}(\Pi'\Pi)}{\pi_{\mathrm{min}}}),
\end{align*}
where we have used the fact that $\frac{\tau+\delta_{\mathrm{max}}}{\tau+\delta_{\mathrm{min}}}=O(1)$ when $\tau\geq M\mathrm{max}(N,J)$ since $\delta_{\mathrm{min}}\leq\delta_{\mathrm{max}}\leq M\mathrm{max}(N,J)$. Then by Condition \ref{Con1} and Lemma \ref{rowwiseerror}, we get
\begin{align*}
\mathrm{max}_{i\in[N]}\|e'_{i}(\hat{\Pi}-\Pi\mathcal{P})\|_{1}=O(\sqrt{\frac{\mathrm{log}(N)}{\rho N}}) \mathrm{~and~}\mathrm{max}_{i\in[N]}\|e'_{i}(\hat{\Pi}_{*}-\Pi\mathcal{P}_{*})\|_{1}=O(\sqrt{\frac{\mathrm{log}(N)}{\rho N}}).
\end{align*}

Next, we bound $\|\hat{\Theta}-\Theta\mathcal{P}\|$. Because $R'\hat{\Pi}(\hat{\Pi}'\hat{\Pi})^{-1}$ is usually almost the same as $\hat{\Theta}$, we use the bound of $\|R'\hat{\Pi}(\hat{\Pi}'\hat{\Pi})^{-1}-\Theta\mathcal{P}\|$ as that of $\|\hat{\Theta}-\Theta\mathcal{P}\|$.
\begin{align}
&\|R'\hat{\Pi}(\hat{\Pi}'\hat{\Pi})^{-1}-\Theta\mathcal{P}\|=\|R'\hat{\Pi}(\hat{\Pi}'\hat{\Pi})^{-1}-\mathscr{R}'\Pi(\Pi'\Pi)^{-1}\mathcal{P}\|\leq\|(R'-\mathscr{R}')\hat{\Pi}(\hat{\Pi}'\hat{\Pi})^{-1}\|+\|\mathscr{R}'(\hat{\Pi}(\hat{\Pi}'\hat{\Pi})^{-1}-\Pi(\Pi'\Pi)^{-1}\mathcal{P})\|\notag\\
&\leq\|R-\mathscr{R}\|\|\hat{\Pi}(\hat{\Pi}'\hat{\Pi})^{-1}\|+\|\mathscr{R}\|\hat{\Pi}(\hat{\Pi}'\hat{\Pi})^{-1}-\Pi(\Pi'\Pi)^{-1}\mathcal{P}\|\leq\frac{\|R-\mathscr{R}\|}{\sigma_{K}(\hat{\Pi})}+\rho\sigma_{1}(\Pi)\sigma_{1}(B)\|\hat{\Pi}(\hat{\Pi}'\hat{\Pi})^{-1}-\Pi(\Pi'\Pi)^{-1}\mathcal{P}\|\notag\\
&=O(\frac{\|R-\mathscr{R}\|}{\sigma_{K}(\Pi)})+\rho\sigma_{1}(\Pi)\sigma_{1}(B)\|\hat{\Pi}(\hat{\Pi}'\hat{\Pi})^{-1}-\Pi(\Pi'\Pi)^{-1}\mathcal{P}\|\label{boundThetahat1},
\end{align}
where we use $\sigma_{K}(\Pi)$ to approximate $\sigma_{K}(\hat{\Pi})$ because $\hat{\Pi}$ is a good approximation of $\Pi$. By the proof of Theorem 1 \citep{qing2023latent}, we know that when Assumption \ref{Assum1} holds, with high probability, we have $\|R-\mathscr{R}\|=O(\sqrt{\rho\mathrm{max}(N,J)\mathrm{log}(N+J)})$, which gives that
\begin{align*}
\|R'\hat{\Pi}(\hat{\Pi}'\hat{\Pi})^{-1}-\Theta\mathcal{P}\|_{F}\leq\sqrt{K}\|R'\hat{\Pi}(\hat{\Pi}'\hat{\Pi})^{-1}-\Theta\mathcal{P}\|=O(\frac{\sqrt{\rho K\mathrm{max}(N,J)\mathrm{log}(N+J)}}{\sigma_{K}(\Pi)}).
\end{align*}

Because $\|\Theta\|_{F}=\rho\|B\|_{F}\geq\rho\sigma_{1}(B)$, we obtain
\begin{align*}
\frac{\|R'\hat{\Pi}(\hat{\Pi}'\hat{\Pi})^{-1}-\Theta\mathcal{P}\|_{F}}{\|\Theta\|_{F}}=O(\sqrt{\frac{K\mathrm{max}(N,J)\mathrm{log}(N+J)}{\rho\sigma^{2}_{K}(\Pi)\sigma^{2}_{1}(B)}}).
\end{align*}
By Condition \ref{Con1}, we have
\begin{align*}
\frac{\|\hat{\Theta}-\Theta\mathcal{P}\|_{F}}{\|\Theta\|_{F}}=O(\sqrt{\frac{\mathrm{log}(N)}{\rho N}}).
\end{align*}
Similarly, we have $\frac{\|\hat{\Theta}_{*}-\Theta\mathcal{P}_{*}\|_{F}}{\|\Theta\|_{F}}=O(\sqrt{\frac{\mathrm{log}(N)}{\rho N}})$.
\end{proof}
\bibliographystyle{elsarticle-num}
\bibliography{refRSCGoM}
\end{document}